\documentclass[aps,prr,twocolumn,10pt,superscriptaddress]{revtex4-1}

\usepackage{mathrsfs}
\usepackage{bbm}
\usepackage{amsmath}
\usepackage{amssymb}
\usepackage{amsthm}
\usepackage{graphicx}
\usepackage{MnSymbol}
\usepackage{mathtools}
\usepackage{bbold}
\usepackage[italicdiff]{physics}
\usepackage{chemformula}

\makeatletter
\usepackage{hyperref}
\usepackage{color}
\definecolor{supcol}{RGB}{10,50,180}
\definecolor{eqcol}{RGB}{220,10,100}
\hypersetup{
	colorlinks,
	citecolor=supcol,
	linkcolor=eqcol,
	urlcolor=supcol
}

\allowdisplaybreaks

\newtheorem{theorem}{Theorem}

\newtheorem{proposition}[theorem]{Proposition}

\DeclareMathOperator{\atanh}{artanh}

\newcommand{\mca}{\mathcal}
\newcommand{\mbb}{\mathbb}

\newcommand{\sectionprl}[1]{{\em #1}\/.---}

\begin{document}
\title{Dissipation, quantum coherence, and asymmetry of finite-time cross-correlations}

\author{Tan Van Vu}
\email{tan.vu@riken.jp}

\affiliation{Analytical quantum complexity RIKEN Hakubi Research Team, RIKEN Center for Quantum Computing (RQC), 2-1 Hirosawa, Wako, Saitama 351-0198, Japan}

\author{Van Tuan Vo}

\affiliation{Department of Physics, Kyoto University, Kyoto 606-8502, Japan}

\author{Keiji Saito}

\affiliation{Department of Physics, Kyoto University, Kyoto 606-8502, Japan}

\date{\today}

\begin{abstract}
Recent studies have revealed a deep connection between the asymmetry of cross-correlations and thermodynamic quantities in the short-time limit. In this study, we address the finite-time domain of the asymmetry for both open classical and quantum systems. Focusing on Markovian dynamics, we show that the asymmetry observed in finite-time cross-correlations is upper bounded by dissipation. We prove that, for classical systems in a steady state with arbitrary operational durations, the asymmetry exhibits, at most, linear growth over time, with the growth speed determined by the rates of entropy production and dynamical activity. In the long-time regime, the asymmetry exhibits exponential decay, with the decay rate determined by the spectral gap of the transition matrix. Remarkably, for quantum cases, quantum coherence is equally important as dissipation in constraining the asymmetry of correlations. We demonstrate an example where only quantum coherence bounds the asymmetry while the entropy production rate vanishes. Furthermore, we generalize the short-time bounds on correlation asymmetry, as reported by Shiraishi [Phys.~Rev.~E 108, L042103 (2023)] and Ohga et al.~[Phys.~Rev.~Lett.~131, 077101 (2023)], to encompass finite-time scenarios. These findings offer novel insights into the thermodynamic aspects of correlation asymmetry.
\end{abstract}

\maketitle

\section{Introduction}
Asymmetry, which refers to the absence of symmetry, constitutes a fundamental concept in physics.
The presence of asymmetry typically results in nontrivial and crucial consequences for a given system and has thus drawn considerable attention in various scientific disciplines.
Nonequilibrium thermodynamics is one such field where the concept of asymmetry is crucial \cite{Sekimoto.2010,Seifert.2012.RPP}.
For instance, the violation of time-reversal symmetry indicates the existence of nonequilibrium conditions, and the degree of such a violation is closely related to entropy production, with asymmetry being reflected by fluctuation theorems \cite{Gallavotti.1995.PRL,Jarzynski.1997.PRL,Kurchan.1998.JPA,Crooks.1999.PRE,Saito.2008.PRB,Esposito.2009.RMP}.
As manifested in the thermodynamic uncertainty relation \cite{Barato.2015.PRL,Gingrich.2016.PRL,Hasegawa.2019.PRL,Timpanaro.2019.PRL,Horowitz.2020.NP} and the entropic bound \cite{Dechant.2018.PRE}, the asymmetry of arbitrary currents is constrained by dissipation.
Another intriguing phenomenon is the relaxation asymmetry, which asserts that the heating process is faster than the cooling one \cite{Lapolla.2020.PRL,Vu.2021.PRR,Manikandan.2021.PRR}.
Furthermore, asymmetry can be harnessed to enhance the performance of heat engines \cite{Mondal.2020.PRE,Harunari.2021.PRR}.

Cross-correlation is a fundamental quantity that embodies both temporal and spatial information pertaining to physical systems.
Fluctuation-dissipation theorems \cite{Kubo.1991} establish that the response of a nonequilibrium steady-state system to a small perturbation can be expressed in terms of cross-correlation \cite{Seifert.2010.EPL}.
The investigation of correlation properties has progressed in various directions, including linear response theory \cite{Crisanti.2003.JPA,Marconi.2008.PR}, speed limits for auto-correlation \cite{Mohan.2022.PRA,Carabba.2022.Q,Hasegawa.2023.arxiv,Dechant.2023.PRL,Mori.2023.PRL}, and thermodynamic inference of entropy production using the correlation between different observables \cite{Dechant.2021.PRX}, to name a few.
Recently, a deep connection between the asymmetry of cross-correlations and thermodynamic quantities has been reported for steady-state systems in the {\it short-time} limit \cite{Ohga.2023.PRL,Shiraishi.2023.PRE}, providing novel insights into the circulation of fluctuation \cite{Tomita.1974.PTP} and coherent oscillations \cite{Barato.2017.PRE}.
From a dynamic standpoint, correlation asymmetry can be related to physical quantities in nonequilibrium systems such as odd viscosity in fluid dynamics \cite{Fruchart.2023.ARCMP}, kinetic  fluxes in reaction networks \cite{Qian.2004.PNAS}, and the temporal ordering of cellular events in living cells \cite{Sisan.2010.BJ}.
Simultaneously, correlation asymmetry should also contain information on the time-reversal symmetry breaking for the whole time regime \cite{Onsager.1931.PR.I,Onsager.1931.PR.II}.
Thus, it is quite important to unveil the in-depth relationship between the asymmetry and the thermodynamic costs in the entire {\it finite-time} domain, as well as to explore quantum effects on correlation.

\begin{figure}[b]
\centering
\includegraphics[width=1\linewidth]{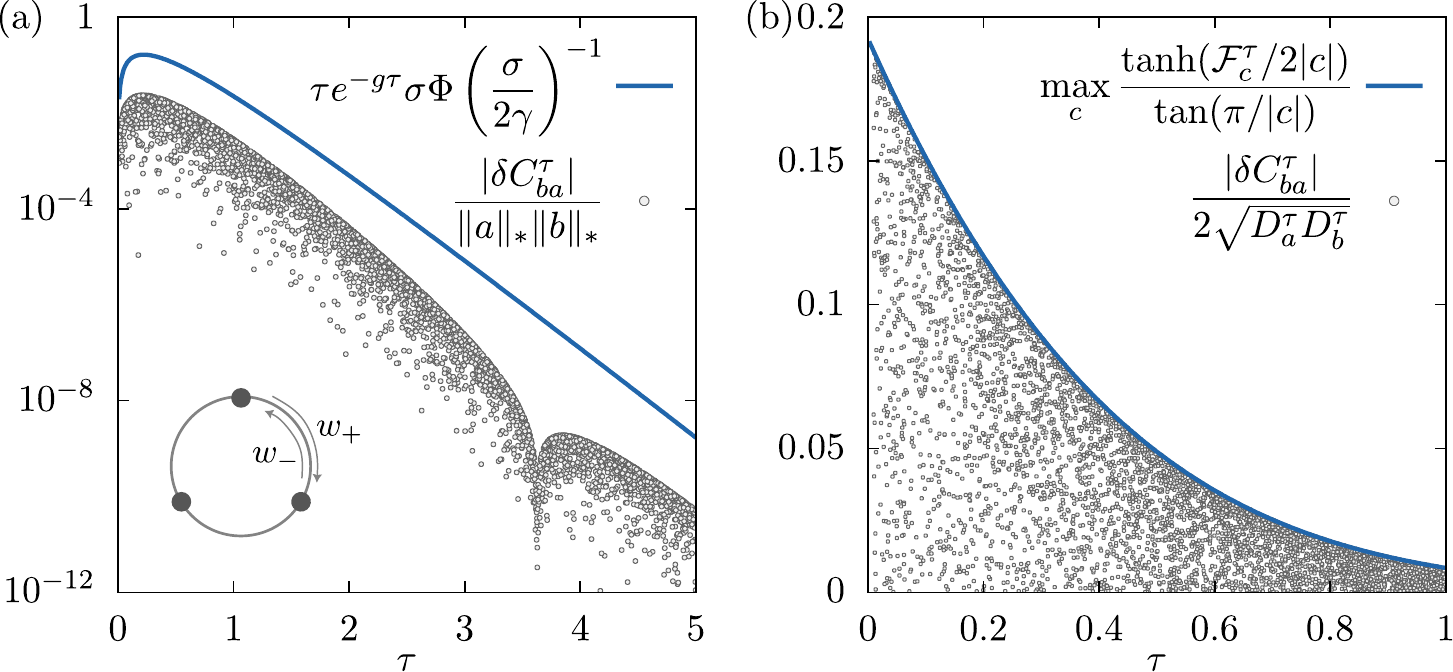}
\protect\caption{Numerical illustration of thermodynamic bounds on the asymmetry of cross-correlations in terms of (a) entropy production [cf.~Eq.~\eqref{eq:asym.cc.fnt.1}] and (b) thermodynamic affinity [cf.~Eq.~\eqref{eq:asym.cc.fnt.aff}] in a three-state biochemical oscillation \cite{Barato.2017.PRE}. Observables $a$ and $b$ are randomly sampled in the range $[-1,1]$, and the forward and backward transition rates are $w_+=2$ and $w_-=1$, respectively.}\label{fig:Cover}
\end{figure}

In the present paper, we address these open problems by considering classical and quantum Markov processes of discrete-state systems whose initial state is a stationary state.
We prove that the asymmetry exhibited in finite-time cross-correlations is always bounded from above by the thermodynamic costs for the whole time regime.
Specifically, for classical Markov jump processes, we show that the asymmetry grows at most linear in time, with the velocity determined by the rates of entropy production and dynamical activity; in the long-time regime, it exponentially decays with the rate determined by the spectral gap of the transition matrix [cf.~Eq.~\eqref{eq:asym.cc.fnt.1}].
These provide a basic picture of the asymmetry in generic Markov processes. 
Variants of this structure are found in other cases with multi-time and multi-observables, as well as in other dynamics such as discrete-time Markov chains, continuous-state overdamped Langevin systems, and open quantum systems. 
Remarkably, by considering the Lindblad master equations for open quantum systems, we find that quantum coherence plays a crucial role in the asymmetry of correlations [cf.~Eq.~\eqref{eq:qua.asym.cc}].
We demonstrate that degeneracy in the spectrum of the steady-state density matrix can lead to a nontrivial phenomenon wherein only quantum coherence is responsible for a finite asymmetry with zero entropy production.
In addition, we provide finite-time generalizations for the short-time bounds reported in Refs.~\cite{Ohga.2023.PRL,Shiraishi.2023.PRE} [cf.~Eqs.~\eqref{eq:asym.cc.fnt.2} and \eqref{eq:asym.cc.fnt.aff}].
These findings provide thermodynamic bounds on physically relevant quantities in various dynamics and further deepen our understanding of the asymmetry of correlations from the thermodynamic perspective.

\section{Setup}
We consider a Markov jump process described by the master equation:
\begin{equation}\label{eq:mas.eq}
    \ket{\dot p_t}={W}\ket{p_t},
\end{equation}
where the dot $\cdot$ denotes the time derivative, $\ket{p_t}=[p_1(t),\dots,p_N(t)]^\top$ is the probability distribution at time $t$, and ${W}=[w_{mn}]\in\mbb{R}^{N\times N}$ denotes the time-independent transition matrix with $w_{mn}\ge 0$ being the jump rate from state $n$ to $m\,(\neq n)$ and $w_{nn}=-\sum_{m(\neq n)}w_{mn}$.
We assume the local detailed balance $\ln(w_{mn}/w_{nm})=\Delta s_{mn}$, i.e., the log of the ratio of transition rates is related to the entropy change in the environment $\Delta s_{mn}$.
Hereafter, we consider the case that the system is in a nonequilibrium steady state $\ket{\pi}$.
Let $j_{mn}\coloneqq w_{mn}\pi_n-w_{nm}\pi_m$ be the steady-state probability current; then, the master equation \eqref{eq:mas.eq} implies $\sum_mj_{nm}=0$.
Two quantities of importance are the rates of entropy production and dynamical activity, defined as
\begin{align}
    \sigma&\coloneqq\sum_{m>n}j_{mn}\ln\frac{w_{mn}\pi_n}{w_{nm}\pi_m},\\
    \gamma&\coloneqq\sum_{m>n}(w_{mn}\pi_n+w_{nm}\pi_m).
\end{align}
Qualitatively, $\sigma$ quantifies the degree of thermodynamic irreversibility, whereas $\gamma$ reflects the timescale of the system \cite{Maes.2020.PR}.
Another relevant quantity is dynamical state mobility \cite{Vu.2023.PRX}, which characterizes the response of probability currents against thermodynamic forces and is defined as
\begin{equation}
	\kappa\coloneqq\sum_{m>n}\frac{j_{mn}}{\ln(w_{mn}\pi_n/w_{nm}\pi_m)}.
\end{equation}
The relation $\kappa\le\gamma/2$ holds in general.

Next, we introduce cross-correlation and some notations to be used in this study.
Let $\ket{a}=[a_1,\dots,a_N]^\top$ and $\ket{b}=[b_1,\dots,b_N]^\top$ be arbitrary observables.
The two-time cross-correlation between these two observables can be defined as
\begin{equation}
    C_{ba}^\tau\coloneqq\ev{b(\tau)a(0)},
\end{equation}
where $a(t)\,[b(t)]$ takes the value of $a_n\,[b_n]$ if the system is in state $n$ at time $t$ and the average $\ev{\cdot}$ is over all stochastic trajectories of time period $\tau$.
Defining $\Pi={\rm diag}(\pi_1,\dots,\pi_N)$, then the cross-correlation can be analytically expressed as $C_{ba}^\tau=\mel{b}{e^{{W}\tau}\Pi}{a}$.
We are interested in the asymmetry of cross-correlations
\begin{equation}
	\delta C_{ba}^\tau\coloneqq C_{ba}^\tau-C_{ab}^\tau,
\end{equation}
which vanishes in equilibrium. However, it is not the case for nonequilibrium situations.
Qualitatively, $\delta C_{ba}^\tau$ reflects symmetry breaking in the causality of observations and can reveal essential aspects of system dynamics, such as the extent to which the system deviates from equilibrium.
Several thermodynamic bounds for this quantity have recently been derived in the $\tau\to 0$ limit \cite{Ohga.2023.PRL,Shiraishi.2023.PRE}.
In this study, we focus on the entire {\it finite-time} regime.

Let $\{\lambda_n\}$ be the set of eigenvalues and $\{\ket{v_n^l},\ket{v_n^r}\}$ be the left and right eigenvectors of ${W}$, respectively (i.e., $\bra{v_n^l}{W}=\bra{v_n^l}\lambda_n$ and ${W}\ket{v_n^r}=\lambda_n\ket{v_n^r}$).
The largest eigenvalue $\lambda_1=0$ is associated with the eigenvector $\bra{v_1^l}\propto\bra{1}$, whereas all other eigenvalues have a negative real part, $0>\Re{\lambda_2}\ge\dots\ge \Re{\lambda_N}$.
Here, $\ket{1}$ denotes the all-one vector.
The eigenvectors are normalized, $\braket{v_n^l}=1$.
Since $\{\ket{v_n^l}\}$ forms a basis of $\mbb{C}^N$, we can always find coefficients $\{\tilde{z}_n\}$ for any vector $\ket{z}$ such that $\ket{z}=\sum_n\tilde{z}_n\ket{v_n^l}$.
Hereafter, we define the $\ell_1$-norm $\|z\|_*\coloneqq\sum_{n\ge 2}|\tilde{z}_n|$.
An important quantity is the spectral gap $g\coloneqq-\Re{\lambda_2}>0$, which corresponds to the slowest decay mode and characterizes the relaxation speed of the system.

\section{Main results}
Given the above setup, we are now ready to explain our results; a simple numerical illustration is presented in Fig.~\ref{fig:Cover}.

\subsection{First main result and quantum generalization}
Our first main result is a thermodynamic bound on the asymmetry of finite-time cross-correlations:
\begin{equation}\label{eq:asym.cc.fnt.1}
    \frac{|\delta C_{ba}^\tau|}{\|a\|_*\|b\|_*}\le\tau e^{-g\tau}\sigma\Phi\qty(\frac{\sigma}{2\gamma})^{-1}\le 2\tau e^{-g\tau}\sqrt{\sigma\kappa},
\end{equation}
where $\Phi(u)$ denotes the inverse function of $u\tanh(u)$ and satisfies $\Phi(u)\ge\max\{u,\sqrt{u}\}\,\forall u\ge 0$.
The proof is presented in Appendix \ref{app:proof.main}. 
Bound \eqref{eq:asym.cc.fnt.1} implies that the thermodynamic costs govern the asymmetry over the whole time domain.
In the small-time regime, the asymmetry of cross-correlations increases at most linear in time with speed constrained by the entropy production and dynamical activity rates.
On the other hand, in the large-time regime, the asymmetry exponentially decays with the rate of the spectral gap.
The result can be analogously generalized to other scenarios, such as discrete-time Markov chains, continuous-state overdamped Langevin systems, and multiple observables, as shown in Appendix \ref{app:gen}.

Some remarks on this finding are given in order.
(i) First, it is of fundamental importance as it relates two physically important concepts (i.e., dissipation and correlation asymmetry) for \emph{arbitrary} times.
While the bound may not be tight, it essentially captures both the dynamic and thermodynamic properties of correlation asymmetry.
An improved saturable bound, which implicitly exhibits the exponential decay, can be found in Appendix \ref{app:improv.bound}.
(ii) Second, notice that $\delta C_{ba}^\tau$ is identical with the difference between cross-correlations in the original dynamics and the dual dynamics \cite{Crooks.2000.PRE} with transition rates $\widetilde{{W}}=\Pi{W}^\dagger\Pi^{-1}$ \cite{fnt1}.
Therefore, the bound \eqref{eq:asym.cc.fnt.1} can also be interpreted as a thermodynamic bound for the discrepancy between these dynamics in terms of observables.
(iii) Third, in other relevant contexts, $\delta C_{ba}^\tau$ can be employed as a natural measure to study interactions in complex systems \cite{Nolte.2006.PRE}, dynamic co-localization, diffusion, binding in living cells \cite{Bacia.2006.NM}, and information flow in biological systems \cite{Ohga.2023.PRL,Horowitz.2014.PRX}.
Thus, the inequality \eqref{eq:asym.cc.fnt.1} not only provides information on the temporal behavior of such a measure but also describes the constraints imposed by the thermodynamic costs.
Particularly, in the case of active fluids, it has been shown that odd viscosity can be expressed as the time integral of correlation asymmetry \cite{Hargus.2020.JCP}.
Our finding thus yields a thermodynamic bound on odd viscosity, indicating that it is constrained by dissipation and the reciprocal of the spectral gap.
(iv) Last, an effective approach to estimating the spectral gap $g$ approximately, without knowing the details of the underlying dynamics, can be deduced from our result. 
In practice, calculating the spectral gap $g$ requires knowledge of the transition matrix, which is generally unavailable in experiments. 
The bound derived here suggests that $g$ can be estimated as the decay slope of the correlation asymmetry, which is experimentally accessible.

Next, we extend the result \eqref{eq:asym.cc.fnt.1} to quantum cases, which include autonomous thermal engines \cite{Mitchison.2019.CP,Manzano.2021.PRL} and quantum measurement processes \cite{Wiseman.2009}.
We consider a Markovian open quantum system, which is weakly coupled to a single or multiple reservoirs.
The time evolution of the reduced density matrix is described by the Lindblad equation, $\dot\varrho_t=\mca{L}(\varrho_t)$ \cite{Lindblad.1976.CMP,Gorini.1976.JMP}, where
\begin{equation}
	\mca{L}(\varrho)\coloneqq -i[H,\varrho]+\sum_k(L_k\varrho L_k^\dagger-\{L_k^\dagger L_k,\varrho\}/2)
\end{equation}
with $H$ is the time-independent Hamiltonian and $\{L_k\}$ denote jump operators.
In order to be thermodynamically consistent, we assume the local detailed balance condition \cite{Horowitz.2013.NJP,Manzano.2022.QS}; that is, the jump operators come in pairs $(k,k')$ such that $L_k=e^{\Delta s_k/2}L_{k'}^\dagger$, where $\Delta s_k$ denotes the entropy change in the environment.
Let $\pi$ be the steady state and $\pi=\sum_n\pi_n\dyad{n}$ be its spectral decomposition.
The rates of irreversible entropy production and dynamical activity are given by $\sigma=\sum_k\tr{L_k\pi L_k^\dagger}\Delta s_k$ and $\gamma=\sum_k\tr{L_k\pi L_k^\dagger}$, respectively.
The system is measured by the eigenbasis $\{\dyad{n}\}$ at both the initial and final times.
Note that this two-point measurement scheme does not alter the steady state of the system.
Define observables ${A}\coloneqq\sum_na_n\dyad{n}$ and ${B}\coloneqq\sum_nb_n\dyad{n}$.
As will be shown later, quantum properties emerge even for this measurement basis.
In this case, the cross-correlation can be analytically calculated as \cite{fnt2}
\begin{equation}
	C_{ba}^\tau=\ev{b(\tau)a(0)}=\tr{{B} e^{\mca{L}\tau}({A}\pi)}.
\end{equation}
Let $\tilde{\mca{L}}$ be an adjoint superoperator of $\mca{L}$, defined as
\begin{equation}
	\tilde{\mca{L}}(\varrho)\coloneqq i[H,\varrho]+\sum_k(L_k^\dagger\varrho L_k-\{L_k^\dagger L_k,\varrho\}/2).
\end{equation}
Note that both $\tilde{\mca{L}}$ and $\mca{L}$ have the same eigenvalue spectrum, $\tilde{\mca{L}}({V}_n)=\lambda_n{V}_n$, where $0=\lambda_1>\Re{\lambda_2}\ge\dots$ and the eigenvectors are normalized such that $\|{V}_n\|_\infty=1$.
For any operator ${X}$, let ${X}_t\coloneqq e^{\tilde{\mca{L}}t}({X})$ be the time-evolved operator in the Heisenberg picture, and $\|X\|_*\coloneqq\sum_{n\ge 2}|z_n^x|$ be the $\ell_1$-norm of $X$, where $X=\sum_nz_n^xV_n$.
The definition of the spectral gap is analogous to the classical case, $g\coloneqq-\Re{\lambda_2}$.

With the above setup in place, we obtain a quantum extension of the classical result \eqref{eq:asym.cc.fnt.1}, indicating that the asymmetry of cross-correlations is limited by both dissipation and quantum coherence (see Appendix \ref{app:qua} for the proof):
\begin{equation}
	|\delta C_{ba}^\tau|\le \tau e^{-g\tau}\qty[\mca{C} + \|A\|_*\|B\|_*\sigma\Phi\qty(\frac{\sigma}{2\gamma})^{-1}].\label{eq:qua.asym.cc}
\end{equation}
Here, $\mca{C}$ is a quantum coherence term that quantifies the amount of quantum coherence generated in the time-evolved observables in the Heisenberg picture, given by
\begin{equation}
	\mca{C}\coloneqq \Theta\int_0^1\dd{s}\qty[ \|B\|_*{C}_{\ell_1}(A_{(1-s)\tau}) + \|A\|_*{C}_{\ell_1}(B_{s\tau}) ],
\end{equation}
where $\Theta\coloneqq 4\|H\|_\infty+3\sum_k\|L_k\|_\infty^2$, $\|\cdot\|_\infty$ denotes the operator norm, and ${C}_{\ell_1}(X_t)\coloneqq e^{gt}\sum_{m\neq n}|\mel{m}{X_t}{n}|$ is the $\ell_1$-norm of quantum coherence in the eigenbasis \cite{Baumgratz.2014.PRL}.
Note that ${C}_{\ell_1}(X_t)$ is always upper bounded by $N(N-1)\|X\|_*$ for all $t\ge 0$.
Roughly speaking, a nonvanishing value of $\mca{C}$ signifies the ability to generate quantum coherence in the original dynamics.
In the classical limit (e.g., $H=0$ and $L_k\propto\dyad{m}{n}$), $\mca{C}$ vanishes.
Bound \eqref{eq:qua.asym.cc} establishes a qualitative and quantitative relationship between asymmetry, thermodynamic irreversibility, and quantum coherence.
Remarkably, $\delta C_{ba}^\tau$ can be nonzero even when $\sigma=0$, indicating that quantum coherence is inevitable in the bound (see Appendix \ref{app:qua.demon} for an analytical analysis).
While this bound may not be strict, it yields a crucial implication that coherent manipulations can be exploited to break the symmetry of correlations even in the absence of dissipation.
Note that $C_{\ell_1}$ defined for operators is different from the conventional quantum coherence in quantum states; nonetheless, they share the same essence of quantum coherence in quantum dynamics.
Specifically, the coherence term $C_{\ell_1}(A_t/B_t)$ is exactly a weighted sum of quantum coherence generated at time $t$ with respect to initial incoherent states $\{\dyad{n}\}$ in the Heisenberg picture \cite{fnt3}.
Due to the importance of the Heisenberg picture in both dynamic and thermodynamic aspects, this kind of quantum coherence is also physically relevant as the conventional one in the Schr{\"o}dinger picture.

\begin{figure}[t]
\centering
\includegraphics[width=1\linewidth]{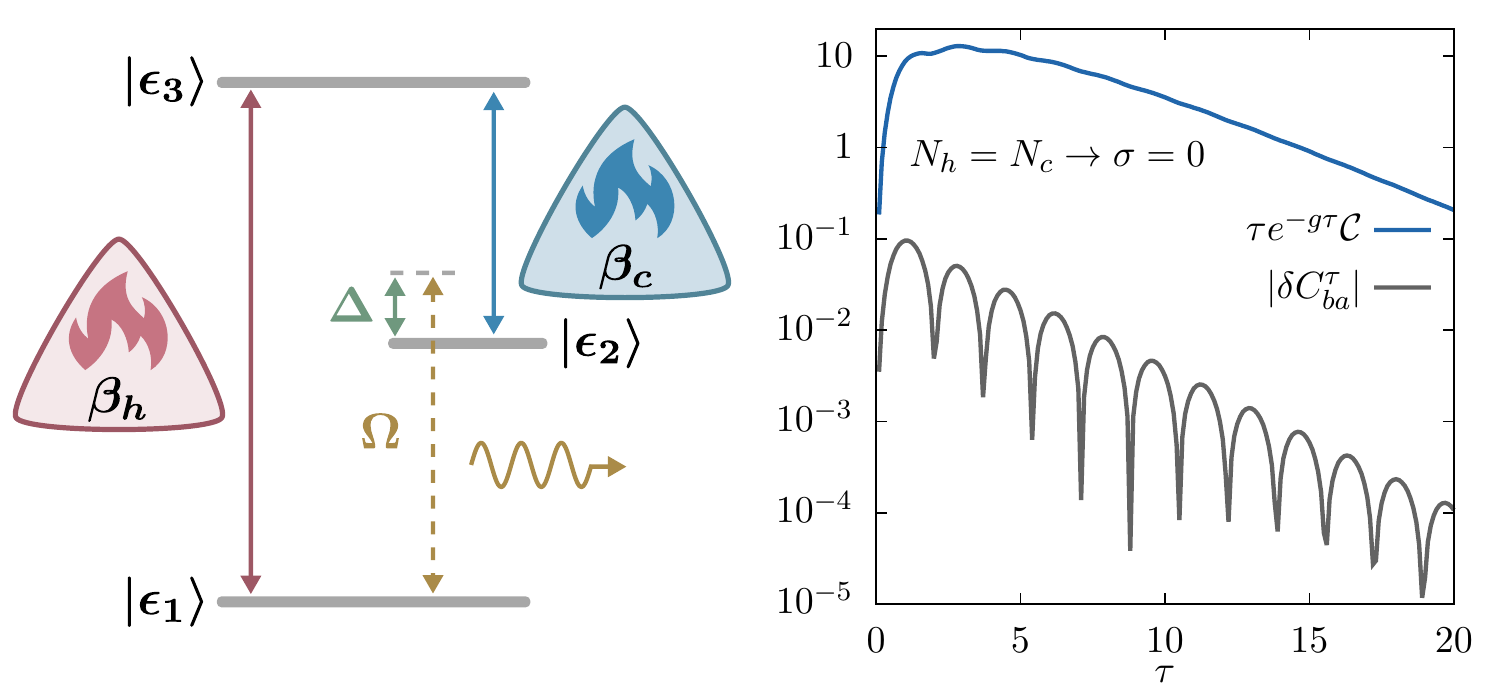}
\protect\caption{Numerical illustration of the inevitable role of quantum coherence in restricting the asymmetry of cross-correlations in the three-level maser. Observables are given by $A=\dyad{1}$ and $B=\dyad{3}$. Parameters are $\Delta=0.01$, $\Omega=1$, $\alpha_h=1$, $\alpha_c=0.01$, $N_h=N_c=1$, $\phi=\pi/2$, and $\theta=\pi/4$. Note that $\sigma =0$. Nevertheless, a finite asymmetry appears.}\label{fig:Coherence}
\end{figure}

We numerically illustrate this critical role of quantum coherence in a three-level maser \cite{Scovil.1959.PRL}, which can operate as a heat engine or refrigerator.
In a rotating frame \cite{Boukobza.2007.PRL,Kalaee.2021.PRE,Vu.2022.PRL.TUR}, the time evolution of the density matrix can be described by the Hamiltonian $H=-\Delta\sigma_{22}+\Omega(\sigma_{12}+\sigma_{21})$ and the jump operators $L_{1}=\sqrt{\alpha_h(N_h+1)}\sigma_{13}$, $L_{1'}= \sqrt{\alpha_hN_h}\sigma_{31}$, $L_{2}=\sqrt{\alpha_c(N_c+1)}\sigma_{23}$, and $L_{2'}=\sqrt{\alpha_cN_c}\sigma_{32}$, where $\sigma_{ij}=\dyad{\epsilon_i}{\epsilon_j}$ (see Appendix \ref{app:qua.demon} for details).
We exclusively consider the $N_h=N_c$ case, in which $\sigma=0$.
Due to degeneracy in the spectrum of $\pi$, the measurement basis can be chosen as $\ket{1}=e^{i\phi}\cos\theta\ket{\epsilon_1}+\sin\theta\ket{\epsilon_2}$, $\ket{2}=-\sin\theta\ket{\epsilon_1}+e^{-i\phi}\cos\theta\ket{\epsilon_2}$, and $\ket{3}=\ket{\epsilon_3}$, where $\phi$ and $\theta$ are arbitrary real numbers.
As confirmed in Fig.~\ref{fig:Coherence}, the asymmetry of correlations does not vanish and is bounded solely by the quantum coherence term $\mca{C}$.
It is worth noting that the origin of this correlation asymmetry arises from external coherent control, which induces degeneracy in the stationary state \cite{fnt4}.
By exploiting this phenomenon through a suitable choice of measurement basis, symmetry breaking of correlations can be achieved.

\subsection{Second and third main results}
In the sequel, we revisit \emph{classical} systems and focus on some quantifications of the normalized asymmetry that have been considered in the literature.
Our second main result is a thermodynamic bound for the normalized asymmetry only in terms of entropy production (see Appendix \ref{app:main2.proof} for the proof):
\begin{equation}\label{eq:asym.cc.fnt.2}
	\frac{|\delta C_{ba}^\tau|^2}{D_a^\tau+D_b^\tau}\le\tau\sigma\min\qty{\|a^2+b^2\|_\infty,\qty[\frac{\max_c\ell_c}{2N\tan(\pi/N)}]^2},
\end{equation}
where $D_a^\tau\coloneqq C_{aa}^0-C_{aa}^\tau$ is the decay of auto-correlation \cite{Ohga.2023.PRL}, $\|a^2+b^2\|_\infty\coloneqq\max_n(a_n^2+b_n^2)$, the maximum is over all cycles of a uniform cycle decomposition \cite{Schnakenberg.1976.RMP,Pietzonka.2016.JPA}, and $\ell_c\coloneqq\sum_{i}\sqrt{(a_{n_i}-a_{n_{i+1}})^2+(b_{n_i}-b_{n_{i+1}})^2}$ is the length of cycle $c=(n_1,\dots,n_{|c|})$.
Here, $|c|$ denotes the size of $c$ and $|c|+1\equiv 1$.
In the short-time limit, bound \eqref{eq:asym.cc.fnt.2} can be reduced to the extant result obtained by Shiraishi \cite{Shiraishi.2023.PRE} [see Eqs.~(6) and (8) therein].
Therefore, it can be regarded as a finite-time generalization of Shiraishi's bound.
Since the normalized asymmetry is experimentally measurable, our bound can be used to estimate entropy production based on trajectory data obtained from experiments, similar to the tool provided by the thermodynamic uncertainty relation \cite{Li.2019.NC,Manikandan.2020.PRL,Vu.2020.PRE,Otsubo.2020.PRE}.
It is noteworthy that this bound is tight and saturable in the short-time limit \cite{Shiraishi.2023.PRE}, while it may not be the case for finite times due to the fact that the asymmetry of cross-correlations decays exponentially, as shown in Eq.~\eqref{eq:asym.cc.fnt.1}.
This leads us to anticipate the existence of exponentially decaying bounds.
Investigating such bounds is left as future work.

So far, we have demonstrated that entropy production serves as a limit for the asymmetry of cross-correlations in finite time.
Our third primary finding is an additional constraint expressed in terms of thermodynamic affinity, given by
\begin{equation}\label{eq:asym.cc.fnt.aff}
\frac{|\delta C_{ba}^\tau|}{2\sqrt{D_a^\tau D_b^\tau}}\le\max_c\frac{\tanh(\mca{F}_c^\tau/2|c|)}{\tan(\pi/|c|)}\le\max_c\frac{\mca{F}_c^\tau}{2\pi},	
\end{equation}
where $\mca{F}_c^\tau$ is the thermodynamic affinity associated with a cycle $c=(n_1,\dots,n_{|c|})$ in temporal coarse-grained dynamics with timescale $\tau$:
\begin{equation}
	\mca{F}_c^\tau\coloneqq\ln\frac{w_{n_2n_1}^\tau w_{n_3n_2}^\tau\dots w_{n_1n_{|c|}}^\tau}{w_{n_1n_2}^\tau w_{n_2n_3}^\tau\dots w_{n_{|c|}n_1}^\tau}.
\end{equation}
Here, $w_{mn}^\tau\coloneqq[e^{{W}\tau}]_{mn}$ is the conditional transition probability from state $n$ to $m$ within time $\tau$.
The proof is presented in Appendix \ref{app:main3.proof}.
As shown in Fig.~\ref{fig:Cover}(b), this bound is tight and can be saturated.
For example, in the three-state biochemical oscillation with homogeneous transition rates, the equality can be attained for an arbitrary time $\tau$ with observables $\ket{a}=[\sin(2\pi n/3)]_n^\top$ and $\ket{b}=[\cos(2\pi n/3)]_n^\top$ (see Appendix \ref{app:main3.proof} for details).
In this case, $|\delta C_{ba}^\tau|$ quantifies coarse-grained oscillation asymmetry by considering only the initial and final times.
In the short- and long-time limits, we can show that $\lim_{\tau\to 0}\mca{F}_c^\tau=\mca{F}_c$ and $\lim_{\tau\to\infty}\mca{F}_c^\tau=0$, where $\mca{F}_c$ denotes the thermodynamic affinity defined for the transition rate matrix \cite{fnt5}.
Therefore, bound \eqref{eq:asym.cc.fnt.aff} can be considered a finite-time generalization of the extant bound reported by Ohga and coworkers \cite{Ohga.2023.PRL}.
Additionally, it provides an estimation of the maximum thermodynamic affinity in temporal coarse-grained dynamics.
This is highly relevant from an experimental point of view due to the limited measurement resolution and may provide new insights into determining dissipative timescales in terms of thermodynamic affinity \cite{Cisneros.2023.JSM}.

\section{Summary and discussion}
We derived finite-time thermodynamic bounds for the asymmetry of cross-correlations in terms of entropy production and thermodynamic affinity.
The results universally apply to various dynamics, from discrete to continuous time and space domains.
At the fundamental level, our findings indicate that dissipation limits the asymmetry of cross-correlations across the entire time regime.
This not only generalizes the principle of microscopic reversibility to nonequilibrium situations but also yields thermodynamic constraints on physical functions, such as viscosity in active fluids and signal transduction in biological systems.
In the context of quantum systems, we have elucidated the pivotal role of quantum coherence in breaking the symmetry of correlations, offering fresh insights into the relationship between asymmetry, dissipation, and quantum coherence.
From a practical standpoint, our results can be applied to infer the spectral gap and dissipative quantities, such as entropy production and coarse-grained thermodynamic affinity.

It is also worthwhile to explore applications of our results in biochemical systems, where the concepts of correlation and symmetry breaking play crucial roles in the performance of systems \cite{Govern.2014.PRL,Cao.2015.NP,Barato.2017.PRE,Seara.2021.NC,Leighton.2022.PRL,Liang.2022.arxiv}.
Furthermore, it would be intriguing to develop our bounds in other quantum scenarios, such as when the system is measured in a basis different from the eigenbasis of the steady state.
Quantum measurements have a unique impact that can influence the asymmetry of correlations, causing them to persist even when the system is initially in equilibriums.
Such investigations could provide valuable insights into the behavior of quantum systems and potentially lead to advancements in harnessing the merits of quantum measurements.

\begin{acknowledgments}
The authors thank Tomotaka Kuwahara for fruitful discussions.
{T.V.V.}~was supported by JSPS KAKENHI Grant Number JP23K13032.
{K.S.}~was supported by JSPS KAKENHI Grant Numbers JP23H01099, JP19H05603, and JP19H05791.
\end{acknowledgments}

\appendix

\onecolumngrid

\section{Proof of Eq.~\eqref{eq:asym.cc.fnt.1} and an improved bound}

\subsection{Proof of Eq.~\eqref{eq:asym.cc.fnt.1}}\label{app:proof.main}
For convenience, we define $J\coloneqq [j_{mn}]$, which is the matrix of probability currents.
By simple algebraic calculations, we can transform the asymmetry of cross-correlations as follows:
\begin{align}
    \delta C_{ba}^\tau&=\mel{b}{e^{{W}\tau}\Pi}{a}-\mel{a}{e^{{W}\tau}\Pi}{b}\notag\\
    &=\mel{b}{e^{{W}\tau}\Pi}{a}-\mel{b}{\Pi e^{{W}^\dagger\tau}}{a}\notag\\
    &=\mel{b}{(e^{{W}\tau}\Pi-\Pi e^{{W}^\dagger\tau})}{a}.\label{eq:asym.cc.tmp0}
\end{align}
Here, we use the fact $z=z^\dagger$ for any $z\in\mbb{R}$ in the second line.
Note that for any matrices ${X}$ and ${Y}$, the following equality always holds:
\begin{equation}
    e^{{X}}\Pi-\Pi e^{{Y}}=\int_0^1\dd{s}e^{s{X}}({X}\Pi-\Pi{Y})e^{(1-s){Y}}.
\end{equation}
Applying ${X}={{W}}\tau$ and ${Y}={{W}}^\dagger\tau$ and noting that ${J} ={{W}}\Pi-\Pi{{W}}^\dagger$, we can proceed Eq.~\eqref{eq:asym.cc.tmp0} further as follows:
\begin{align}
    \delta C_{ba}^\tau&=\tau\int_0^1\dd{s}\mel{b}{e^{s{{W}}\tau}{J} e^{(1-s){{W}}^\dagger\tau}}{a}\notag\\
    &=\tau\int_0^1\dd{s}\tr{{J} e^{(1-s){{W}}^\dagger\tau}\dyad{a}{b}e^{s{{W}}\tau}}.\label{eq:asym.cc.tmp1}
\end{align}
Noting that $\ket{a}=\sum_n\tilde{a}_n\ket{v_n^l}$, $\ket{b}=\sum_n\tilde{b}_n\ket{v_n^l}$, and $\int_0^1\dd{s}e^{(1-s)x+sy}=(e^x-e^y)/(x-y)$, we can simplify the terms in Eq.~\eqref{eq:asym.cc.tmp1} as
\begin{align}
    \delta C_{ba}^\tau&=\tau\int_0^1\dd{s}\sum_{m,n}\tr{{J}\tilde{a}_m\tilde{b}_n^*e^{(1-s)\lambda_m^*\tau+s\lambda_n\tau}\dyad{v_m^l}{v_n^l}}\notag\\
    &=\tau\tr{{J}\sum_{m,n\ge 2}\tilde{a}_m\tilde{b}_n^*\frac{e^{\lambda_m^*\tau}-e^{\lambda_n\tau}}{(\lambda_m^*-\lambda_n)\tau}\dyad{v_m^l}{v_n^l}},\label{eq:asym.cc.comp}
\end{align}
where we use the facts that $\bra{v_1^l}{J}=\bra{0}$ and ${J}\ket{v_1^l}=\ket{0}$ to obtain the last line.
Here, $\ket{0}$ denotes the all-zero vector.
Before proceeding further, we note some useful inequalities (see Appendix \ref{app:useful.ines} for the proof). 
For any matrices ${X}$ and ${Y}$ and orthogonal basis $\{\ket{n}\}$, we have $\qty|\tr{{X}{Y}}|\le\|{Y}\|_\infty\sum_{m,n}|\mel{m}{{X}}{n}|$, where $\|{Y}\|_\infty$ denotes the operator norm of ${Y}$.
In addition, for any complex numbers $z_1$ and $z_2$ with a negative real part (i.e., $\Re{z_1}\le 0$ and $\Re{z_2}\le 0$), we always have $|e^{z_1}-e^{z_2}|/|z_1-z_2|\le 1$.
Using the expression \eqref{eq:asym.cc.comp} and applying the above inequalities, we can evaluate as follows:
\begin{align}
    |\delta C_{ba}^\tau|&\le\tau\sum_{m,n}|j_{mn}|\left\|\sum_{m,n\ge 2}\tilde{a}_m\tilde{b}_n^*\frac{e^{\lambda_m^*\tau}-e^{\lambda_n\tau}}{(\lambda_m^*-\lambda_n)\tau}\dyad{v_m^l}{v_n^l}\right\|_\infty\notag\\
    &\le\tau\sum_{m,n}|j_{mn}|\sum_{m,n\ge 2}|\tilde{a}_m\tilde{b}_n^*|\qty|\frac{e^{\lambda_m^*\tau}-e^{\lambda_n\tau}}{(\lambda_m^*-\lambda_n)\tau}|.\label{eq:asym.cc.tmp2}
\end{align}
Next, we evaluate the last term in Eq.~\eqref{eq:asym.cc.tmp2}.
For $m\ge n\ge 2$, since $\Re{\lambda_m^*-\lambda_n}\le 0$ and $\Re{\lambda_n}\le-g$, we thus have 
\begin{align}
    \qty|\frac{e^{\lambda_m^*\tau}-e^{\lambda_n\tau}}{(\lambda_m^*-\lambda_n)\tau}|&=|e^{\lambda_n\tau}|\qty|\frac{e^{(\lambda_m^*-\lambda_n)\tau}-1}{(\lambda_m^*-\lambda_n)\tau}|\le e^{-g\tau}.
\end{align}
Likewise, for $n\ge m\ge 2$, since $\Re{\lambda_n-\lambda_m^*}\le 0$ and $\Re{\lambda_m^*}\le-g$, we also obtain the same argument.
Applying these inequalities to Eq.~\eqref{eq:asym.cc.tmp2}, we readily obtain
\begin{equation}\label{eq:asym.cc.tmp3}
    |\delta C_{ba}^\tau|\le \tau e^{-g\tau}\|a\|_*\|b\|_*\sum_{m,n}|j_{mn}|.
\end{equation}
Furthermore, we can prove that the sum of absolute probability currents is bounded from above by the rates of entropy production, dynamical activity, and dynamical state mobility as \cite{Vu.2023.PRX}
\begin{equation}
    \sum_{m,n}|j_{mn}|\le\sigma\Phi\qty(\frac{\sigma}{2\gamma})^{-1}\le 2\sqrt{\sigma\kappa}.\label{eq:asym.cc.tmp4}
\end{equation}
Combining Eqs.~\eqref{eq:asym.cc.tmp3} and \eqref{eq:asym.cc.tmp4} yields the desired result \eqref{eq:asym.cc.fnt.1} in the main text.

\subsection{A quantitative improvement of bound \eqref{eq:asym.cc.fnt.1} and numerical demonstration}\label{app:improv.bound}

Here we demonstrate that the bound \eqref{eq:asym.cc.fnt.1} can be quantitatively improved; however, this improvement comes at the cost of \emph{implicitly} exhibiting the exponential decay.
To this end, we define $\ket{x(t)}\coloneqq (e^{W^\dagger t}-\dyad{1}{\pi})\ket{x}$ for $x\in\{a,b\}$.
Using this notation and Eq.~\eqref{eq:asym.cc.tmp1}, the correlation asymmetry can be calculated as
\begin{align}
    \delta C_{ba}^\tau&=\tau\int_0^1\dd{s}\mel{b(s\tau)}{J}{a((1-s)\tau)}\notag\\
    &=\tau\sum_{m\neq n}\int_0^1\dd{s}j_{mn}a_n((1-s)\tau)b_m(s\tau)\notag\\
    &=\tau\sum_{m>n}\int_0^1\dd{s}j_{mn}\qty[a_n((1-s)\tau)b_m(s\tau)-a_m((1-s)\tau)b_n(s\tau)].
\end{align}
Defining $\gamma_{mn}\coloneqq w_{mn}\pi_n+w_{nm}\pi_m$ and applying the Cauchy-Schwarz inequality, we can evaluate the asymmetry from above as
\begin{align}
	|\delta C_{ba}^\tau|&\le \tau\int_0^1\dd{s}\sqrt{\sum_{m>n}\frac{j_{mn}^2}{\gamma_{mn}}\sum_{m>n}\qty[a_n((1-s)\tau)b_m(s\tau)-a_m((1-s)\tau)b_n(s\tau)]^2\gamma_{mn}}\notag\\
	&=\tau\sqrt{\sum_{m>n}\frac{j_{mn}^2}{\gamma_{mn}}}\int_0^1\dd{s}\sqrt{\sum_{m>n}\qty[a_n((1-s)\tau)b_m(s\tau)-a_m((1-s)\tau)b_n(s\tau)]^2\gamma_{mn}}.\label{eq:asym.cc.ptmp1}
\end{align}
Note that the first term in Eq.~\eqref{eq:asym.cc.ptmp1} is the pseudo entropy production rate and can be upper bounded as \cite{Vo.2022.JPA}
\begin{equation}
	\sum_{m>n}\frac{j_{mn}^2}{\gamma_{mn}}\le \frac{\sigma^2}{4\gamma}\Phi\qty(\frac{\sigma}{2\gamma})^{-2}.
\end{equation}
Defining 
\begin{equation}
	D_{ba}^\tau\coloneqq \int_0^1\dd{s}\sqrt{\sum_{m>n}\qty[a_n((1-s)\tau)b_m(s\tau)-a_m((1-s)\tau)b_n(s\tau)]^2\gamma_{mn}},\label{eq:Dabt.def}
\end{equation}
we immediately obtain an improvement for Eq.~\eqref{eq:asym.cc.fnt.1} as
\begin{equation}
	|\delta C_{ba}^\tau|\le \tau D_{ba}^\tau\frac{\sigma}{2\sqrt{\gamma}}\Phi\qty(\frac{\sigma}{2\gamma})^{-1},\label{eq:improved.bound}
\end{equation}
where $D_{ba}^\tau$ exponentially decays over time.
Notably, this improved bound is tight and can be saturated, for example, in a three-state biochemical oscillation (see Fig.~\ref{fig:Demo} for the numerical demonstration). 
We can also prove that
\begin{equation}
	D_{ba}^\tau\le 2e^{-g\tau}\|a\|_*\|b\|_*\sqrt{\gamma}.\label{eq:Dabt.ub}
\end{equation}
Using this relation, the bound \eqref{eq:asym.cc.fnt.1} in the main text can be recovered.

\begin{figure}[b]
\centering
\includegraphics[width=0.9\linewidth]{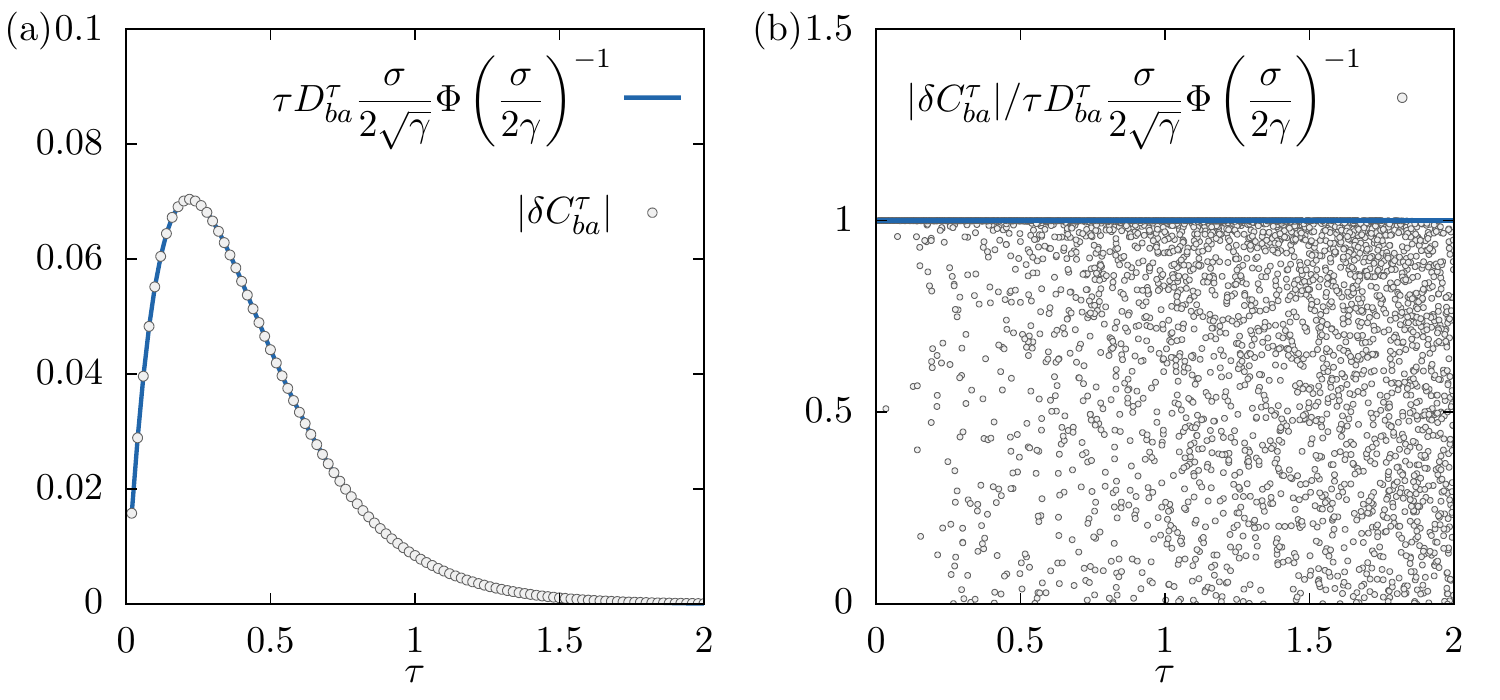}
\protect\caption{Numerical illustration of the improved bound \eqref{eq:improved.bound} in a three-state biochemical oscillation. The forward and backward transition rates are $w_+=2$ and $w_-=1$, respectively. (a) Observables are fixed as $\ket{a}=[\sin(2\pi n/3)]_n^\top$ and $\ket{b}=[\cos(2\pi n/3)]_n^\top$ while time $\tau$ is varied. (b) Observables $a$ and $b$ are randomly sampled in the range $[-1,1]$ for each random time $\tau\in[0,2]$.}\label{fig:Demo}
\end{figure} 

\sectionprl{Proof of Eq.~\eqref{eq:Dabt.ub}}Note that
\begin{equation}
	e^{W^\dagger t}=\dyad{1}{\pi}+\sum_{n\ge 2}e^{\lambda_n^*t}\dyad{v_n^l}{v_n^r}.
\end{equation}
We begin by upper bounding the term inside the integral in Eq.~\eqref{eq:Dabt.def} as follows:
\begin{align}
	&\sum_{m>n}\qty[a_n((1-s)\tau)b_m(s\tau)-a_m((1-s)\tau)b_n(s\tau)]^2\gamma_{mn}\notag\\
	&\le 4\|a((1-s)\tau)\|_\infty^2\|b(s\tau)\|_\infty^2\sum_{m>n}\gamma_{mn}\\
	&=4\gamma\|\sum_{n\ge 2}e^{\lambda_n^*(1-s)\tau}\ket{v_n^l}\braket{v_n^r}{a}\|_\infty^2\|\sum_{n\ge 2}e^{\lambda_n^*s\tau}\ket{v_n^l}\braket{v_n^r}{b}\|_\infty^2\notag\\
	&=4\gamma\|\sum_{n\ge 2}e^{\lambda_n^*(1-s)\tau}\tilde{a}_n\ket{v_n^l}\|_\infty^2\|\sum_{n\ge 2}e^{\lambda_n^*s\tau}\tilde{b}_n\ket{v_n^l}\|_\infty^2\notag\\
	&\le 4\gamma e^{-2g\tau}\qty(\sum_{n\ge 2}\|e^{(\lambda_n^*+g)(1-s)\tau}\tilde{a}_n\ket{v_n^l}\|_\infty)^2\qty(\sum_{n\ge 2}\|e^{(\lambda_n^*+g)s\tau}\tilde{b}_n\ket{v_n^l}\|_\infty)^2\notag\\
	&\le 4\gamma e^{-2g\tau}\|a\|_*^2\|b\|_*^2.
\end{align}
Consequently, $D_{ba}^\tau$ is upper bounded as
\begin{equation}
	D_{ba}^\tau\le \int_0^1\dd{s}2e^{-g\tau}\|a\|_*\|b\|_*\sqrt{\gamma}=2e^{-g\tau}\|a\|_*\|b\|_*\sqrt{\gamma}.
\end{equation}

\section{Generalizations of Eq.~\eqref{eq:asym.cc.fnt.1} to other cases}\label{app:gen}
\subsection{Generalization to discrete-time Markov chains}
Here we provide a generalization of Eq.~\eqref{eq:asym.cc.fnt.1} for discrete-time Markov chains.
We consider a time-homogeneous irreducible Markov chain whose dynamics is governed by the master equation:
\begin{equation}
	\ket{p_{t_i}}={R}\ket{p_{t_{i-1}}}.
\end{equation}
Here, ${R}=[R_{mn}]$ is the stochastic matrix with $R_{mn}\ge 0$ the transition probability from state $n$ to $m$. 
The normalization condition $\sum_mR_{mn}=1$ is satisfied for all $n$.
We consider a finite duration $\tau$ of $K$ steps (i.e., $t_0=0$ and $t_K=\tau$).
The system is in a nonequilibrium steady state $\ket{\pi}$ (i.e., $\ket{\pi}={R}\ket{\pi}$).
The steady-state entropy production and dynamical activity at each time step are given by
\begin{align}
	\sigma&\coloneqq\sum_{m\neq n}R_{mn}\pi_n\ln\frac{R_{mn}\pi_n}{R_{nm}\pi_m},\\
	\gamma&\coloneqq\sum_{m\neq n}R_{mn}\pi_n.
\end{align}
The cross-correlation between the observables can be expressed as
\begin{equation}
	C_{ba}^\tau\coloneqq\ev{b(\tau)a(0)}=\mel{b}{{R}^K\Pi}{a}.
\end{equation}
Let $\{\lambda_n\}$ be the set of eigenvalues of ${R}$ and $\{\ket{v_n^l},\ket{v_n^r}\}$ be the left and right eigenvectors, respectively (i.e., $\bra{v_n^l}{R}=\bra{v_n^l}\lambda_n$ and ${R}\ket{v_n^r}=\lambda_n\ket{v_n^r}$).
Note that $1=\lambda_1>|\lambda_2|\ge\dots$ and $\ket{v_1}\propto\ket{1}$.
The spectral gap can thus be defined as $g\coloneqq -\ln|\lambda_2|>0$.
Evidently, $|\lambda_n|\le e^{-g}$ for any $n\ge 2$.

Now, we can calculate the asymmetry of cross-correlations as follows:
\begin{align}
	\delta C_{ba}^\tau&=\mel{b}{{R}^K\Pi}{a}-\mel{a}{{R}^K\Pi}{b}\notag\\
	&=\mel{b}{{R}^K\Pi - \Pi({R}^\dagger)^K}{a}.
\end{align}
Note that for any matrices ${X}$ and ${Y}$, the following relation holds:
\begin{equation}
{X}^K\Pi-\Pi{Y}^K=\sum_{k=0}^{K-1}{X}^{k}({X}\Pi-\Pi{Y}){Y}^{K-1-k}.	
\end{equation}
Applying this relation for ${X}={R}$ and ${Y}={R}^\dagger$ and noting that ${J}={R}\Pi-\Pi{R}^\dagger$, we can proceed further as
\begin{align}
	\delta C_{ba}^\tau&=\mel{b}{{R}^K\Pi - \Pi({R}^\dagger)^K}{a}\notag\\
	&=\sum_{k=0}^{K-1}\mel{b}{{R}^k({R}\Pi-\Pi{R}^\dagger)({R}^\dagger)^{K-1-k}}{a}\notag\\
	&=\sum_{k=0}^{K-1}\tr{{J}({R}^\dagger)^{K-1-k}\dyad{a}{b}{R}^k}\notag\\
	&=\sum_{k=0}^{K-1}\sum_{m,n}\tr{{J} a_mb_n^*(\lambda_m^*)^{K-1-k}\lambda_n^{k}\dyad{v_m^l}{v_n^l}}.
\end{align} 
Note that ${J}\ket{1}=\ket{0}$ and $\bra{1}{J}=\bra{0}$.
Consequently, the asymmetry of cross-correlations can be upper bounded as
\begin{align}
	|\delta C_{ba}^\tau|&=\qty|\sum_{k=0}^{K-1}\sum_{m,n\ge 2}\tr{{J} a_mb_n^*(\lambda_m^*)^{K-1-k}\lambda_n^{k}\dyad{v_m^l}{v_n^l}}|\notag\\
	&\le \sum_{m,n}|j_{mn}|\left\|\sum_{k=0}^{K-1}\sum_{m,n\ge 2}a_mb_n^*(\lambda_m^*)^{K-1-k}\lambda_n^{k}\dyad{v_m^l}{v_n^l}\right\|_\infty\notag\\
	&\le Ke^{-(K-1)g}\sum_{m,n}|j_{mn}|\sum_{m,n\ge 2}|a_mb_n^*|\notag\\
	&\le Ke^{-(K-1)g}\|a\|_*\|b\|_*\sigma\Phi\qty(\frac{\sigma}{2\gamma})^{-1},
\end{align}
which yields the desired generalization for discrete-time systems:
\begin{equation}\label{eq:asym.cc.dis.fnt}
    \frac{|\delta C_{ba}^{\tau}|}{\|a\|_*\|b\|_*}\le K e^{-(K-1)g}\sigma\Phi\qty(\frac{\sigma}{2\gamma})^{-1}.
\end{equation}

\subsection{Generalization to overdamped Langevin dynamics}
We consider a $d$-dimensional overdamped Langevin system.
Let $p_t(\vb*{x})$ denote the probability density function of finding the system in state $\vb*{x}$ at time $t$.
The time evolution of $p_t(\vb*{x})$ is described by the Fokker-Planck equation:
\begin{equation}
\dot p_t(\vb*{x})=\mca{L}[p_t(\vb*{x})]=-\vb*{\nabla}\cdot\vb*{j}_t(\vb*{x}),
\end{equation}
where $\mca{L}[p(\vb*{x})]\coloneqq -\vb*{\nabla}\cdot[\vb*{f}(\vb*{x})p(\vb*{x})-{D}\vb*{\nabla}p(\vb*{x})]$ is the Fokker-Planck operator, $\vb*{f}(\vb*{x})$ is the force vector, and ${D}={\rm diag}(D_1,\dots,D_d)$ is the matrix of diffusion coefficients.
Consider the adjoint operator $\tilde{\mca{L}}$, which is defined as
\begin{equation}
\tilde{\mca{L}}[p(\vb*{x})]\coloneqq	\vb*{f}(\vb*{x})\cdot\vb*{\nabla}p(\vb*{x})+\vb*{\nabla}\cdot{D}\vb*{\nabla}p(\vb*{x}).
\end{equation}
The operator $\tilde{\mca{L}}$ is also known as the generator of the backward Fokker-Planck equation.
Define the inner product
\begin{equation}
	\ev{u(\vb*{x}),v(\vb*{x})}\coloneqq\int\dd{\vb*{x}} u(\vb*{x})v(\vb*{x}).
\end{equation}
For any functions $u(\vb*{x})$ and $v(\vb*{x})$ such that $u(\vb*{x})\vb*{f}(\vb*{x})v(\vb*{x})$, $u(\vb*{x}){D}\vb*{\nabla}v(\vb*{x})$, and $v(\vb*{x}){D}\vb*{\nabla}u(\vb*{x})$ vanish at infinity, we prove that
\begin{equation}
	\ev{u(\vb*{x}),\mca{L}[v(\vb*{x})]}=\ev{v(\vb*{x}),\tilde{\mca{L}}[u(\vb*{x})]}.
\end{equation}
Indeed, by exploiting the boundary conditions, we can show as
\begin{align}
	\ev{u(\vb*{x}),\mca{L}[v(\vb*{x})]}&=-\int\dd{\vb*{x}} u(\vb*{x})\vb*{\nabla}\cdot[\vb*{f}(\vb*{x})v(\vb*{x})-{D}\vb*{\nabla}v(\vb*{x})]\notag\\
	&=\int\dd{\vb*{x}}[\vb*{f}(\vb*{x})v(\vb*{x})-{D}\vb*{\nabla}v(\vb*{x})]\cdot\vb*{\nabla}u(\vb*{x})\notag\\
	&=\int\dd{\vb*{x}}[v(\vb*{x})\vb*{f}(\vb*{x})\cdot\vb*{\nabla}u(\vb*{x})+v(\vb*{x})\vb*{\nabla}\cdot{D}\vb*{\nabla}u(\vb*{x})]\notag\\
	&=\int\dd{\vb*{x}} v(\vb*{x})\qty[\vb*{f}(\vb*{x})\cdot\vb*{\nabla}u(\vb*{x})+\vb*{\nabla}\cdot{D}\vb*{\nabla}u(\vb*{x})]\notag\\
	&=\ev{v(\vb*{x}),\tilde{\mca{L}}[u(\vb*{x})]}.
\end{align}
Using this relation, we can immediately derive that
\begin{equation}
	\ev{u(\vb*{x}),e^{\mca{L}\tau}[v(\vb*{x})]}=\ev{v(\vb*{x}),e^{\tilde{\mca{L}}\tau}[u(\vb*{x})]}.\label{eq:ovd.adj}
\end{equation}
We consider the case where the system is in a nonequilibrium steady state $\pi(\vb*{x})$ (i.e., $\mca{L}[\pi(\vb*{x})]=-\vb*{\nabla}\cdot\vb*{j}(\vb*{x})=0$).
The entropy production rate can be calculated as
\begin{equation}
	\sigma=\int\dd{\vb*{x}}\frac{\vb*{j}(\vb*{x})\cdot{D}^{-1}\vb*{j}(\vb*{x})}{\pi(\vb*{x})}.
\end{equation}
The cross-correlation between observables $a(\vb*{x})$ and $b(\vb*{x})$ can be expressed as
\begin{equation}
	C_{ba}^\tau\coloneqq\ev{b(\vb*{x}_\tau)a(\vb*{x}_0)}=\int\dd{\vb*{x}} b(\vb*{x})e^{\mca{L}\tau}\qty[a(\vb*{x})\pi(\vb*{x})].
\end{equation}
Using the equality \eqref{eq:ovd.adj}, the asymmetry of cross-correlations can thus be calculated as
\begin{align}
	\delta C_{ba}^\tau&=\int\dd{\vb*{x}} b(\vb*{x})e^{\mca{L}\tau}\qty[a(\vb*{x})\pi(\vb*{x})]-\int\dd{\vb*{x}} a(\vb*{x})e^{\mca{L}\tau}\qty[b(\vb*{x})\pi(\vb*{x})]\notag\\
	&=\int\dd{\vb*{x}} b(\vb*{x})\qty{ e^{\mca{L}\tau}[a(\vb*{x})\pi(\vb*{x})] - \pi(\vb*{x})e^{\tilde{\mca{L}}\tau}[a(\vb*{x})] }\notag\\
	&=\int\dd{\vb*{x}} b(\vb*{x})\int_0^1\dd{s}\frac{d}{ds}{ e^{s\mca{L}\tau}\qty{e^{(1-s)\tilde{\mca{L}}\tau}[a(\vb*{x})]\pi(\vb*{x})} }\notag\\
	&=\tau\int\dd{\vb*{x}} b(\vb*{x})\int_0^1\dd{s} e^{s\mca{L}\tau}\qty{\mca{L}[\pi(\vb*{x})\circ] - \pi(\vb*{x})\tilde{\mca{L}}[\circ]}\qty[e^{(1-s)\tilde{\mca{L}}\tau}a(\vb*{x})].\label{eq:ovd.tmp4}
\end{align}
From the stationarity, we also obtain that
\begin{align}
	&\mca{L}[\pi(\vb*{x})q(\vb*{x})] - \pi(\vb*{x})\tilde{\mca{L}}[q(\vb*{x})]\notag\\
	&=-\vb*{\nabla}\cdot\qty{\vb*{f}(\vb*{x})\pi(\vb*{x})q(\vb*{x})-{D}\vb*{\nabla}[\pi(\vb*{x})q(\vb*{x})]} - \pi(\vb*{x})\qty[\vb*{f}(\vb*{x})\cdot\vb*{\nabla}q(\vb*{x})+\vb*{\nabla}\cdot{D}\vb*{\nabla}q(\vb*{x})]\notag\\
	&=-2\vb*{f}(\vb*{x})\pi(\vb*{x})\cdot\vb*{\nabla}q(\vb*{x})-q(\vb*{x})\vb*{\nabla}\cdot[\vb*{f}(\vb*{x})\pi(\vb*{x})]+\vb*{\nabla}\cdot{D}\vb*{\nabla}[\pi(\vb*{x})q(\vb*{x})]-\pi(\vb*{x})\vb*{\nabla}\cdot{D}\vb*{\nabla}q(\vb*{x})\notag\\
	&=-2\vb*{f}(\vb*{x})\pi(\vb*{x})\cdot\vb*{\nabla}q(\vb*{x})-q(\vb*{x})\vb*{\nabla}\cdot[\vb*{f}(\vb*{x})\pi(\vb*{x})]+\vb*{\nabla}\pi(\vb*{x})\cdot{D}\vb*{\nabla}q(\vb*{x})+\vb*{\nabla}q(\vb*{x})\cdot{D}\vb*{\nabla}\pi(\vb*{x})+q(\vb*{x})\vb*{\nabla}\cdot{D}\vb*{\nabla}\pi(\vb*{x})\notag\\
	&=-2\vb*{\nabla}q(\vb*{x})\cdot[\vb*{f}(\vb*{x})\pi(\vb*{x})-{D}\vb*{\nabla}\pi(\vb*{x})]-q(\vb*{x})\vb*{\nabla}\cdot[\vb*{f}(\vb*{x})\pi(\vb*{x})-{D}\vb*{\nabla}\pi(\vb*{x})]\notag\\
	&=-2\vb*{\nabla}q(\vb*{x})\cdot\vb*{j}(\vb*{x}).\label{eq:ovd.tmp5}
\end{align}
Let $\{\lambda_n\}$ be the discrete spectrum of operator $\tilde{\mca{L}}$ and $\{\varphi_n(\vb*{x})\}$ be the set of corresponding eigenfunctions:
\begin{align}
	\tilde{\mca{L}}\varphi_n(\vb*{x})=\lambda_n\varphi_n(\vb*{x}).
\end{align}
Assume that $a(\vb*{x})$ and $b(\vb*{x})$ can be expanded in terms of the eigenfunctions as
\begin{align}
	a(\vb*{x})&=\sum_{n}a_n\varphi_n(\vb*{x}),\\
	b(\vb*{x})&=\sum_{n}b_n\varphi_n(\vb*{x}).
\end{align}
The eigenvalue $\lambda_1=0$ corresponds to the eigenfunction $\varphi_1(\vb*{x})=1$.
Other eigenvalues have a negative real part, $0>\Re{\lambda_2}\ge \Re{\lambda_3}\ge \dots$; thus, the spectral gap can be defined as $g\coloneqq -\Re{\lambda_2}>0$. 
Using Eqs.~\eqref{eq:ovd.tmp4} and \eqref{eq:ovd.tmp5}, we can evaluate the asymmetry of cross-correlations as
\begin{align}
	\delta C_{ba}^\tau&=-\tau\int_0^1\dd{s}\int\dd{\vb*{x}} b(\vb*{x}) e^{s\mca{L}\tau}\qty[\qty{2\vb*{\nabla}e^{(1-s)\tilde{\mca{L}}\tau}a(\vb*{x})}\cdot\vb*{j}(\vb*{x})]\notag\\
	&=-\tau\int_0^1\dd{s}\int\dd{\vb*{x}}\qty[\qty{2\vb*{\nabla}e^{(1-s)\tilde{\mca{L}}\tau}a(\vb*{x})}\cdot\vb*{j}(\vb*{x})]e^{s\tilde{\mca{L}}\tau}[b(\vb*{x})]\notag\\
	&=-\tau\int_0^1\dd{s}\int\dd{\vb*{x}}\qty[\qty{\sum_n2\vb*{\nabla}e^{(1-s)\tau\lambda_n} a_n\varphi_n(\vb*{x})}\cdot\vb*{j}(\vb*{x})]\sum_me^{s\tau\lambda_m}b_m\varphi_m(\vb*{x})\notag\\
	&=-\tau\int\dd{\vb*{x}}\int_0^1\dd{s}\sum_{m,n}e^{s\tau\lambda_m+(1-s)\tau\lambda_n}b_ma_n\qty[2\varphi_m(\vb*{x})\vb*{\nabla}\varphi_n(\vb*{x})]\cdot\vb*{j}(\vb*{x})\notag\\
	&=-\tau\int\dd{\vb*{x}}\sum_{m,n}\frac{e^{\lambda_m\tau}-e^{\lambda_n\tau}}{(\lambda_m-\lambda_n)\tau}b_ma_n\qty[2\varphi_m(\vb*{x})\vb*{\nabla}\varphi_n(\vb*{x})]\cdot\vb*{j}(\vb*{x})\notag\\
	&=-\tau\int\dd{\vb*{x}}\sum_{m,n\ge 2}\frac{e^{\lambda_m\tau}-e^{\lambda_n\tau}}{(\lambda_m-\lambda_n)\tau}b_ma_n\qty[2\varphi_m(\vb*{x})\vb*{\nabla}\varphi_n(\vb*{x})]\cdot\vb*{j}(\vb*{x})\notag\\
	&=\tau\int\dd{\vb*{x}}\sum_{m,n\ge 2}\frac{e^{\lambda_m\tau}-e^{\lambda_n\tau}}{(\lambda_m-\lambda_n)\tau}b_ma_n\qty[\varphi_n(\vb*{x})\vb*{\nabla}\varphi_m(\vb*{x})-\varphi_m(\vb*{x})\vb*{\nabla}\varphi_n(\vb*{x})]\cdot\vb*{j}(\vb*{x})\notag\\
	&=\tau e^{-g\tau}\int\dd{\vb*{x}}\vb*{\varphi}_\tau(\vb*{x})\cdot\vb*{j}(\vb*{x}),
\end{align}
where we use the fact $\int\dd{\vb*{x}}[\varphi_1(\vb*{x})\vb*{\nabla}\varphi_n(\vb*{x})]\cdot\vb*{j}(\vb*{x})=\int\dd{\vb*{x}}[\varphi_m(\vb*{x})\vb*{\nabla}\varphi_1(\vb*{x})]\cdot\vb*{j}(\vb*{x})=0$ and define
\begin{equation}
	\vb*{\varphi}_\tau(\vb*{x})\coloneqq\sum_{m,n\ge 2}\frac{e^{(\lambda_m+g)\tau}-e^{(\lambda_n+g)\tau}}{(\lambda_m-\lambda_n)\tau}b_ma_n\qty[\varphi_n(\vb*{x})\vb*{\nabla}\varphi_m(\vb*{x})-\varphi_m(\vb*{x})\vb*{\nabla}\varphi_n(\vb*{x})].
\end{equation}
Applying the Cauchy-Schwarz inequality, we obtain
\begin{align}
	|\delta C_{ba}^\tau|&\le\tau e^{-g\tau}\qty[\int\dd{\vb*{x}}\frac{\vb*{j}(\vb*{x})\cdot{D}^{-1}\vb*{j}(\vb*{x})}{\pi(\vb*{x})}]^{1/2}\qty[\int\dd{\vb*{x}}\pi(\vb*{x})\vb*{\varphi}_\tau(\vb*{x})\cdot{D}\vb*{\varphi}_\tau(\vb*{x})]^{1/2}\notag\\
	&=\tau e^{-g\tau}\sqrt{\sigma}\qty[\int\dd{\vb*{x}}\pi(\vb*{x})\vb*{\varphi}_\tau(\vb*{x})\cdot{D}\vb*{\varphi}_\tau(\vb*{x})]^{1/2}.\label{eq:ovd.tmp1}
\end{align}
For any vector $\vb*{z}=[z_1,\dots,z_d]^\top$, define $|\vb*{z}|\coloneqq[|z_1|,\dots,|z_d|]^\top$.
Since
\begin{equation}
	\qty| \frac{e^{(\lambda_m+g)\tau}-e^{(\lambda_n+g)\tau}}{(\lambda_m-\lambda_n)\tau} |\le 1~\forall m,n\ge 2,
\end{equation}
we can upper bound the last term in Eq.~\eqref{eq:ovd.tmp1} as
\begin{align}
	\int\dd{\vb*{x}}\pi(\vb*{x})\vb*{\varphi}_\tau(\vb*{x})\cdot{D}\vb*{\varphi}_\tau(\vb*{x})\le \int\dd{\vb*{x}}\pi(\vb*{x})\vb*{\phi}(\vb*{x})\cdot{D}\vb*{\phi}(\vb*{x}),
\end{align}
where we define
\begin{equation}
	\vb*{\phi}(\vb*{x})\coloneqq\sum_{m,n\ge 2}|b_m||a_n||\varphi_n(\vb*{x})\vb*{\nabla}\varphi_m(\vb*{x})-\varphi_m(\vb*{x})\vb*{\nabla}\varphi_n(\vb*{x})|.
\end{equation}
Defining $\chi_{ba}\coloneqq\qty[\int\dd{\vb*{x}}\pi(\vb*{x})\vb*{\phi}(\vb*{x})\cdot{D}\vb*{\phi}(\vb*{x})]^{1/2}$, we obtain the following thermodynamic bound on the asymmetry of cross-correlations:
\begin{align}
	\frac{|\delta C_{ba}^\tau|}{\chi_{ba}}&\le\tau e^{-g\tau}\sqrt{\sigma}.\label{eq:ovd.tmp2}
\end{align}

\subsection{Generalization to multiple observables}
The result \eqref{eq:asym.cc.fnt.1} can be generalized to the case of multi-time and multi-observables, where the observables $(\ket{o^1},\dots,\ket{o^M})\eqqcolon\vb*{o}$ are respectively measured at times $(t_1,\dots,t_M)\eqqcolon\vb*{\tau}$.
Here, $M\ge 2$ is an arbitrary integer number and $0=t_1<t_2<\dots<t_M=\tau$.
In this case, the cross-correlation can be defined as 
\begin{equation}
	C_{\vb*{o}}^{\vb*{\tau}}\coloneqq\ev{o^1(t_1)\dots o^M(t_M)},
\end{equation}
where $o^m(t)$ takes the value of $o^m_n$ if the system is in state $n$ at time $t$.
Define the reversed observation times $\tilde{\vb*{\tau}}\coloneqq(\tau-t_1,\dots,\tau-t_M)$.
Then, the following asymmetry of cross-correlations is a quantity of interest:
\begin{equation}
	\delta C_{\vb*{o}}^{\vb*{\tau}}\coloneqq C_{\vb*{o}}^{\vb*{\tau}}-C_{\vb*{o}}^{\tilde{\vb*{\tau}}}.
\end{equation}
We find that this asymmetry is consistently limited by dissipation and decreases exponentially at the rate of the spectral gap:
\begin{equation}
	\frac{|\delta C_{\vb*{o}}^{\vb*{\tau}}|}{\chi_{\vb*{o}}}\le\sigma\Phi\qty(\frac{\sigma}{2\gamma})^{-1}\sum_{k=2}^M(t_{k}-t_{k-1})e^{-g(t_{k}-t_{k-1})}.\label{eq:asym.cc.fnt.mul.obs}
\end{equation}
Here, $\chi_{\vb*{o}}\coloneqq\|o^1\|_*\|o^M\|_*\prod_{m=2}^{M-1}\|o^m\|_\infty\qty(\sum_{n=1}^N\|v_n^r\|_2)^{M-2}$, $\|z\|_2\coloneqq\sqrt{\braket{z}}$ denotes $\ell_2$-norm, and $\|z\|_\infty=\max_n|z_n|$ for any vector $\ket{z}$.
Interestingly, bound \eqref{eq:asym.cc.fnt.mul.obs} indicates that the degree of decay also depends on the time interval between consecutive measurements.

In what follows, we present the proof of Eq.~\eqref{eq:asym.cc.fnt.mul.obs}.
For convenience, we define ${O}_m\coloneqq{\rm diag}(o^m_1,\dots,o^m_N)$.
Then, $C_{\vb*{o}}^{\vb*{\tau}}$ can be explicitly expressed as
\begin{equation}
	C_{\vb*{o}}^{\vb*{\tau}}=\mel{o^M}{\prod_{k=M-1}^{1}e^{{W}(t_{k+1}-t_k)}{O}_k}{\pi}=\mel{o^M}{e^{{W}(t_{M}-t_{M-1})}{O}_{M-1}\dots{O}_2e^{{W}(t_{2}-t_{1})}\Pi}{o^1}.
\end{equation}
Consequently, the asymmetry can be calculated as follows:
\begin{align}
	\delta C_{\vb*{o}}^{\vb*{\tau}}&=\mel{o^M}{e^{{W}(t_{M}-t_{M-1})}{O}_{M-1}\dots{O}_{2}e^{{W}(t_{2}-t_{1})}\Pi}{o^1}-\mel{o^1}{e^{{W}(t_{2}-t_{1})}{O}_{2}\dots{O}_{M-1}e^{{W}(t_{M}-t_{M-1})}\Pi}{o^M}\notag\\
	&=\mel{o^M}{e^{{W}(t_{M}-t_{M-1})}{O}_{M-1}\dots{O}_{2}e^{{W}(t_{2}-t_{1})}\Pi}{o^1}-\mel{o^M}{\Pi e^{{W}^\dagger(t_{M}-t_{M-1})}{O}_{M-1}\dots{O}_{2}e^{{W}^\dagger(t_{2}-t_{1})}}{o^1}\notag\\
	&=\mel{o^M}{e^{{W}(t_{M}-t_{M-1})}{O}_{M-1}\dots{O}_{2}e^{{W}(t_{2}-t_{1})}\Pi-\Pi e^{{W}^\dagger(t_{M}-t_{M-1})}{O}_{M-1}\dots{O}_{2}e^{{W}^\dagger(t_{2}-t_{1})}}{o^1}.
\end{align}
For simplicity, we define $\Delta_m\coloneqq t_{m}-t_{m-1}>0$.
Noticing that $[{O}_k,\Pi]=0$ for all $k$ and $e^{{W} t}=\sum_{i}e^{\lambda_it}\dyad{v_i^r}{v_i^l}$, we can further transform $\delta C_{\vb*{o}}^{\vb*{\tau}}$ as
\begin{align}
	\delta C_{\vb*{o}}^{\vb*{\tau}}&=\mel{o^M}{\sum_{k=2}^{M}e^{{W}\Delta_{M}}{O}_{M-1}\dots e^{{W}\Delta_{k+1}}{O}_{k}\qty[e^{{W}\Delta_{k}}\Pi-\Pi e^{{W}^\dagger\Delta_{k}}]{O}_{k-1}e^{{W}^\dagger\Delta_{k-1}}\dots{O}_{2}e^{{W}^\dagger\Delta_{2}}}{o^1}\notag\\
	&=\sum_{k=2}^{M}\Delta_k\int_0^1\dd{s}\mel{o^M}{\prod_{m=M-1}^{k} e^{{W}\Delta_{m+1}}{O}_{m}\qty[e^{s{W}\Delta_{k}}{J} e^{(1-s){W}^\dagger\Delta_{k}}]\prod_{m=k-1}^{2}{O}_{m}e^{{W}^\dagger\Delta_{m}}}{o^1}\notag\\
	&=\sum_{k=2}^{M}\Delta_k\sum_{i,j}\frac{e^{\lambda_i\Delta_k}-e^{\lambda_j^*\Delta_k}}{(\lambda_i-\lambda_j^*)\Delta_k}\mel{o^M}{\prod_{m=M-1}^{k} e^{{W}\Delta_{m+1}}{O}_{m}}{v_i^r}\mel{v_i^l}{{J}}{v_j^l}\mel{v_j^r}{\prod_{m=k-1}^{2}{O}_{m}e^{{W}^\dagger\Delta_{m}}}{o^1}\notag\\
	&=\sum_{k=2}^{M}\Delta_k\sum_{i,j\ge 2}\frac{e^{\lambda_i\Delta_k}-e^{\lambda_j^*\Delta_k}}{(\lambda_i-\lambda_j^*)\Delta_k}\mel{o^M}{\prod_{m=M-1}^{k} e^{{W}\Delta_{m+1}}{O}_{m}}{v_i^r}\mel{v_i^l}{{J}}{v_j^l}\mel{v_j^r}{\prod_{m=k-1}^{2}{O}_{m}e^{{W}^\dagger\Delta_{m}}}{o^1}.\label{eq:mul.obs.tmp1}
\end{align}
Next, we upper bound the terms in Eq.~\eqref{eq:mul.obs.tmp1}.
Note that for any $i,j\ge 2$, the following inequality holds:
\begin{equation}
	\qty|\frac{e^{\lambda_i\Delta_k}-e^{\lambda_j^*\Delta_k}}{(\lambda_i-\lambda_j^*)\Delta_k}|\le e^{-g\Delta_k}.\label{eq:mul.obs.tmp5}
\end{equation}
In addition, since $\Re{\lambda_n}\le 0$, we can evaluate as follows:
\begin{align}
	|\mel{o^M}{\prod_{m=M-1}^{k} e^{{W}\Delta_{m+1}}{O}_{m}}{v_i^r}|&=\qty|\sum_{1\le i_{M-1},\dots,i_k\le N}e^{\sum_{m=M-1}^k\lambda_{i_m}\Delta_{m+1}}\braket{o^M}{v_{i_{M-1}}^r}\prod_{m=M-1}^{k}\mel{v_{i_m}^l}{{O}_{m}}{v_{i_{m-1}}^r}|\notag\\
	&\le\sum_{1\le i_{M-1},\dots,i_k\le N}\qty|\braket{o^M}{v_{i_{M-1}}^r}\prod_{m=M-1}^{k}\mel{v_{i_m}^l}{{O}_{m}}{v_{i_{m-1}}^r}|\notag\\
	&\le\sum_{1\le i_{M-1},\dots,i_k\le N}|\braket{o^M}{v_{i_{M-1}}^r}|\prod_{m=M-1}^{k}\|{O}_m\|_\infty\|v_{i_m}^l\|_2\|v_{i_{m-1}}^r\|_2\notag\\
	&=\prod_{m=M-1}^{k}\|{O}_m\|_\infty\qty(\sum_{1\le i_{M-1},\dots,i_k\le N}|\tilde{o}_{i_{M-1}}^M|\prod_{m=M-1}^{k}\|v_{i_{m-1}}^r\|_2)\notag\\
	&\le\|o^M\|_*\qty(\sum_{n=1}^N\|v_n^r\|_2)^{M-1-k}\|v_{i}^r\|_2\prod_{m=M-1}^{k}\|{O}_m\|_\infty.\label{eq:mul.obs.tmp2}
\end{align}
where $i_{k-1}\equiv i$.
Likewise, we can also obtain
\begin{align}
	|\mel{v_j^r}{\prod_{m=k-1}^{2}{O}_{m}e^{{W}^\dagger(t_{m}-t_{m-1})}}{o^1}|&\le \|o^1\|_*\qty(\sum_{n=1}^N\|v_n^r\|_2)^{k-3}\|v_{j}^r\|_2\prod_{m=k-1}^{2}\|{O}_m\|_\infty.\label{eq:mul.obs.tmp3}
\end{align}
By combining Eqs.~\eqref{eq:mul.obs.tmp2} and \eqref{eq:mul.obs.tmp3}, we arrive at the following inequality:
\begin{align}
	\sum_{i,j\ge 2}|\mel{o^M}{\prod_{m=M-1}^{k} e^{{W}(t_{m+1}-t_{m})}{O}_{m}}{v_i^r}\mel{v_j^r}{\prod_{m=k-1}^{2}{O}_{m}e^{{W}^\dagger(t_{m}-t_{m-1})}}{o^1}|\le \|o^1\|_*\|o^M\|_*\prod_{m=2}^{M-1}\|{O}_m\|_\infty\qty(\sum_{n=1}^N\|v_n^r\|_2)^{M-2}.\label{eq:mul.obs.tmp4}
\end{align}
Noting that $|\mel{v_i^l}{{J}}{v_j^l}|\le\sum_{m,n}|j_{mn}|\le\sigma\Phi\qty(\sigma/2\gamma)^{-1}$ and using Eqs.~\eqref{eq:mul.obs.tmp5} and \eqref{eq:mul.obs.tmp4}, we obtain
\begin{equation}
	|\delta C_{\vb*{o}}^{\vb*{\tau}}|\le \|o^1\|_*\|o^M\|_*\prod_{m=2}^{M-1}\|{O}_m\|_\infty\qty(\sum_{n=1}^N\|v_n^r\|_2)^{M-2}\sum_{k=2}^M(t_k-t_{k-1})e^{-g(t_k-t_{k-1})}\sigma\Phi\qty(\sigma/2\gamma)^{-1}.
\end{equation}
Defining $\chi_{\vb*{o}}\coloneqq\|o^1\|_*\|o^M\|_*\prod_{m=2}^{M-1}\|{O}_m\|_\infty\qty(\sum_{n=1}^N\|v_n^r\|_2)^{M-2}$, the generalization of Eq.~\eqref{eq:asym.cc.fnt.1} can be attained as
\begin{equation}
	\frac{|\delta C_{\vb*{o}}^{\vb*{\tau}}|}{\chi_{\vb*{o}}}\le\sigma\Phi\qty(\frac{\sigma}{2\gamma})^{-1}\sum_{k=2}^M(t_{k}-t_{k-1})e^{-g(t_{k}-t_{k-1})}.
\end{equation}

\section{Proof of Eq.~\eqref{eq:qua.asym.cc} and an analytical demonstration}

\subsection{Proof of Eq.~\eqref{eq:qua.asym.cc}}\label{app:qua}
For convenience, we define $w_{mn}^k\coloneqq|\mel{m}{L_k}{n}|^2$ and $j_{mn}^k\coloneqq w_{mn}^k\pi_n-w_{nm}^{k'}\pi_m$.
Using these terms, the rates of entropy production and dynamical activity can be expressed as \cite{Vu.2023.PRX}
\begin{align}
	\sigma&=\frac{1}{2}\sum_k\sum_{n,m}j_{mn}^k\ln\frac{w_{mn}^k\pi_n}{w_{nm}^{k'}\pi_m},\\
	\gamma&=\frac{1}{2}\sum_k\sum_{n,m}(w_{mn}^k\pi_n+w_{nm}^{k'}\pi_m).
\end{align}
Furthermore, it was proved that \cite{Vu.2023.PRX}
\begin{equation}
	\sum_{k}\sum_{n,m}|j_{mn}^k|\le \sigma\Phi\qty(\frac{\sigma}{2\gamma})^{-1},
\end{equation}
which will be used later.
By simple algebraic calculations, it can be shown that the relation $\ev{{X},\mca{L}({Y})}=\ev{\tilde{\mca{L}}({X}),{Y}}$ holds for any operators ${X}$ and ${Y}$, where $\ev{X,Y}\coloneqq\tr{X^\dagger Y}$.
As a consequence, we can prove that
\begin{equation}
	\ev{{X},e^{\mca{L}t}({Y})}=\ev{e^{\tilde{\mca{L}}t}({X}),{Y}}
\end{equation}
for any $t\ge 0$.
Noting that ${A}$ and ${B}$ are self-adjoint operators, we can calculate the asymmetry as follows:
\begin{align}
	\delta C_{ba}^\tau&=\ev{{B},e^{\mca{L}\tau}({A}\pi)}-\ev{{A},e^{\mca{L}\tau}({B}\pi)}\notag\\
	&=\ev{{B},e^{\mca{L}\tau}({A}\pi)}-\ev{e^{\tilde{\mca{L}}\tau}({A}),{B}\pi}\notag\\
	&=\ev{{B},e^{\mca{L}\tau}({A}\pi)}-\ev{{B},\pi e^{\tilde{\mca{L}}\tau}({A})}\notag\\
	&=\ev{{B},e^{\mca{L}\tau}({A}\pi)-\pi e^{\tilde{\mca{L}}\tau}({A})}.
\end{align}
Since $[{A},\pi]=0$, the following equality holds:
\begin{align}
	e^{\mca{L}\tau}({A}\pi)-\pi e^{\tilde{\mca{L}}\tau}({A})&=\int_0^1\dd{s}\frac{d}{ds}\qty[ e^{s\mca{L}\tau}(\pi e^{(1-s)\tilde{\mca{L}}\tau}({A})) ]\notag\\
	&=\tau\int_0^1\dd{s}e^{s\mca{L}\tau}\qty[ \mca{L} (\pi e^{(1-s)\tilde{\mca{L}}\tau}({A})) - \pi\tilde{\mca{L}}(e^{(1-s)\tilde{\mca{L}}\tau}({A})) ].
\end{align}
Using this equality, we can proceed further as
\begin{align}
	\delta C_{ba}^\tau&=\tau\int_0^1\dd{s}\ev{{B},e^{s\mca{L}\tau}\qty[ \mca{L} (\pi e^{(1-s)\tilde{\mca{L}}\tau}({A})) - \pi\tilde{\mca{L}}(e^{(1-s)\tilde{\mca{L}}\tau}({A})) ]}\notag\\
	&=\tau\int_0^1\dd{s}\ev{e^{s\tilde{\mca{L}}\tau}({B}),\mca{L} (\pi e^{(1-s)\tilde{\mca{L}}\tau}({A})) - \pi\tilde{\mca{L}}(e^{(1-s)\tilde{\mca{L}}\tau}({A}))}\notag\\
	&=\tau\int_0^1\dd{s}\ev{{B}_{s\tau},\mca{L} (\pi{A}_{(1-s)\tau}) - \pi\tilde{\mca{L}}({A}_{(1-s)\tau})}.
\end{align}
For any operator ${X}$, we can express ${X}$ in terms of $\{{V}_n\}$ as ${X}=\sum_nz_n^x{V}_n$, where $\{z_n^x\}$ are complex coefficients.
Note that $V_1=\mbb{1}$ and ${X}_{t}=\sum_nz_n^xe^{\lambda_nt}{V}_n$.
Define $\bar{{X}}_t\coloneqq{X}_t-z_1^x\mbb{1}=\sum_{n\ge 2}z_n^xe^{\lambda_nt}{V}_n$.
Since $\ev{\mbb{1},\mca{L} (\pi{A}_{t}) - \pi\tilde{\mca{L}}({A}_{t})}=\ev{{B}_t,\mca{L} (\pi\mbb{1}) - \pi\tilde{\mca{L}}(\mbb{1})}=0$, the asymmetry of cross-correlations can be written as
\begin{align}
	\delta C_{ba}^\tau&=\tau\int_0^1\dd{s}\ev{\bar{{B}}_{s\tau},\mca{L} (\pi\bar{{A}}_{(1-s)\tau}) - \pi\tilde{\mca{L}}(\bar{{A}}_{(1-s)\tau})}.\label{eq:qbound.tmp8}
\end{align}
The term inside the integral can be expressed as
\begin{align}
	\ev{\bar{{B}}_{s\tau},\mca{L} (\pi\bar{{A}}_{(1-s)\tau}) - \pi\tilde{\mca{L}}(\bar{{A}}_{(1-s)\tau})}&=\ev{\bar{{B}}_{s\tau},-i[H,\pi\bar{{A}}_{(1-s)\tau}]-i\pi[H,\bar{{A}}_{(1-s)\tau}]}\notag\\
	&+\sum_k\ev{\bar{{B}}_{s\tau},L_k\pi\bar{{A}}_{(1-s)\tau}L_k^\dagger-\pi L_k^\dagger\bar{{A}}_{(1-s)\tau}L_k+(\pi L_k^\dagger L_kA_{(1-s)\tau}-L_k^\dagger L_k\pi\bar{{A}}_{(1-s)\tau})/2}.\label{eq:qbound.tmp1}
\end{align}
We individually evaluate the terms in Eq.~\eqref{eq:qbound.tmp1}.
The first term can be upper bounded as follows:
\begin{align}
	&\qty|\ev{\bar{{B}}_{s\tau},-i[H,\pi\bar{{A}}_{(1-s)\tau}]-i\pi[H,\bar{{A}}_{(1-s)\tau}]}|\notag\\
	&\le \qty|\ev{\bar{{B}}_{s\tau},[H,\pi\bar{{A}}_{(1-s)\tau}]}|+\qty|\ev{\bar{{B}}_{s\tau},\pi[H,\bar{{A}}_{(1-s)\tau}]}|\notag\\
	&=\Big|\sum_{n}\pi_n\mel{n}{\bar{{A}}_{(1-s)\tau}(\bar{{B}}_{s\tau}H-H\bar{{B}}_{s\tau})}{n}\Big|+\Big|\sum_{n}\pi_n\mel{n}{(H\bar{{A}}_{(1-s)\tau}-\bar{{A}}_{(1-s)\tau}H)\bar{{B}}_{s\tau}}{n}\Big|.\label{eq:qbound.tmp4}
\end{align}
Moreover, the first term in Eq.~\eqref{eq:qbound.tmp4} can be bounded as follows:
\begin{align}
	&\Big|\sum_{n}\pi_n\mel{n}{\bar{{A}}_{(1-s)\tau}(\bar{{B}}_{s\tau}H-H\bar{{B}}_{s\tau})}{n}\Big|\notag\\
	&\le \Big|\sum_{n,m(\neq n)}\pi_n\mel{n}{\bar{{A}}_{(1-s)\tau}}{m}\mel{m}{\bar{{B}}_{s\tau}H-H\bar{{B}}_{s\tau}}{n}\Big|+\Big|\sum_{n}\pi_n\mel{n}{\bar{{A}}_{(1-s)\tau}}{n}\mel{n}{\bar{{B}}_{s\tau}H-H\bar{{B}}_{s\tau}}{n}\Big|\notag\\
	&\le\|\bar{{B}}_{s\tau}H-H\bar{{B}}_{s\tau}\|_\infty\sum_{n,m(\neq n)}\pi_n|\mel{n}{\bar{{A}}_{(1-s)\tau}}{m}|+\Big|\sum_{n,m(\neq n)}\pi_n\mel{n}{\bar{{A}}_{(1-s)\tau}}{n}\qty(\mel{n}{\bar{{B}}_{s\tau}}{m}\mel{m}{H}{n}-\mel{n}{H}{m}\mel{m}{\bar{{B}}_{s\tau}}{n})\Big|\notag\\
	&\le 2\|\bar{{B}}_{s\tau}\|_\infty\|H\|_\infty\sum_{n\neq m}|\mel{n}{\bar{{A}}_{(1-s)\tau}}{m}|+2\|\bar{{A}}_{(1-s)\tau}\|_\infty\|H\|_\infty\sum_{n\neq m}|\mel{n}{\bar{{B}}_{s\tau}}{m}|\notag\\
	&\le 2e^{-g\tau}\|H\|_\infty\qty[\|B\|_*C_{\ell_1}(A_{(1-s)\tau}) + \|A\|_*C_{\ell_1}(B_{s\tau})].\label{eq:qbound.tmp3}
\end{align}
Here, we use the facts that
\begin{align}
	\sum_{n\neq m}|\mel{n}{\bar{X}_{t}}{m}|&=\sum_{n\neq m}|\mel{n}{X_t-z_1^x\mbb{1}}{m}|=\sum_{n\neq m}|\mel{n}{X_t}{m}|=e^{-gt}C_{\ell_1}(X_t),\\
	\|\bar{X}_{t}\|_\infty&=\|\sum_{n\ge 2}z_n^xe^{\lambda_nt}{V}_n\|_\infty\le \sum_{n\ge 2}|z_n^x||e^{\lambda_nt}|\|{V}_n\|_\infty\le e^{-gt}\sum_{n\ge 2}|z_n^x|=e^{-gt}\|X\|_*.
\end{align}
It is worth noting that $C_{\ell_1}(X_t)$ is upper bounded by a constant for all times:
\begin{align}
	C_{\ell_1}(X_t)&=\sum_{n\neq m}|\mel{n}{\sum_{k\ge 2}z_k^xe^{(\lambda_k+g)t}V_k}{m}|\le \sum_{n\neq m}\sum_{k\ge 2}|z_k^x|\|V_k\|_\infty=N(N-1)\|X\|_*.
\end{align}
Therefore, the last quantity in Eq.~\eqref{eq:qbound.tmp3} always decays exponentially at the rate of the spectral gap $g$.
Likewise, we also obtain the following bound for the second term in Eq.~\eqref{eq:qbound.tmp4}:
\begin{align}
	\Big|\sum_{n}\pi_n\mel{n}{(H\bar{{A}}_{(1-s)\tau}-\bar{{A}}_{(1-s)\tau}H)\bar{{B}}_{s\tau}}{n}\Big|\le 2e^{-g\tau}\|H\|_\infty\qty[\|B\|_*C_{\ell_1}(A_{(1-s)\tau}) + \|A\|_*C_{\ell_1}(B_{s\tau})].\label{eq:qbound.tmp5}
\end{align}
Consequently, by combining Eqs.~\eqref{eq:qbound.tmp4}, \eqref{eq:qbound.tmp3}, and \eqref{eq:qbound.tmp5}, we arrive at the following upper bound of the first term in Eq.~\eqref{eq:qbound.tmp1}:
\begin{equation}
	\qty|\ev{\bar{{B}}_{s\tau},-i[H,\pi\bar{{A}}_{(1-s)\tau}]-i\pi[H,\bar{{A}}_{(1-s)\tau}]}|\le 4e^{-g\tau}\|H\|_\infty\qty[\|B\|_*C_{\ell_1}(A_{(1-s)\tau}) + \|A\|_*C_{\ell_1}(B_{s\tau})].\label{eq:qbound.tmp7}
\end{equation}
Now, it remains to evaluate the second term in Eq.~\eqref{eq:qbound.tmp1}, which can be calculated as follows:
\begin{align}
	&\sum_k\ev{\bar{{B}}_{s\tau},L_k\pi\bar{{A}}_{(1-s)\tau}L_k^\dagger-\pi L_k^\dagger\bar{{A}}_{(1-s)\tau}L_k+(\pi L_k^\dagger L_kA_{(1-s)\tau}-L_k^\dagger L_k\pi\bar{{A}}_{(1-s)\tau})/2}\notag\\
	&=\sum_k\sum_{n}\pi_n\mel{n}{\bar{{A}}_{(1-s)\tau}L_k^\dagger\bar{{B}}_{s\tau}L_k-L_k^\dagger\bar{{A}}_{(1-s)\tau}L_k\bar{{B}}_{s\tau}+(L_k^\dagger L_k\bar{{A}}_{(1-s)\tau}\bar{{B}}_{s\tau}-\bar{{A}}_{(1-s)\tau}\bar{{B}}_{s\tau}L_k^\dagger L_k)/2}{n}\notag\\
	&=\sum_k\sum_{n,m,n',m'}\pi_n\qty(\mel{n}{\bar{{A}}_{(1-s)\tau}}{n'}\mel{n'}{L_k^\dagger}{m'}\mel{m'}{\bar{{B}}_{s\tau}}{m}\mel{m}{L_k}{n}-\mel{n}{L_k^\dagger}{m}\mel{m}{\bar{{A}}_{(1-s)\tau}}{m'}\mel{m'}{L_k}{n'}\mel{n'}{\bar{{B}}_{s\tau}}{n})\notag\\
	&+\sum_k\sum_{n,m,n'}\pi_n\qty(\mel{n}{L_k^\dagger}{m}\mel{m}{L_k}{n'}\mel{n'}{\bar{{A}}_{(1-s)\tau}\bar{{B}}_{s\tau}}{n}-\mel{n}{\bar{{A}}_{(1-s)\tau}\bar{{B}}_{s\tau}}{n'}\mel{n'}{L_k^\dagger}{m}\mel{m}{L_k}{n})/2.
\end{align}
Collecting the terms that contain only the diagonal elements of $\bar{{A}}_{(1-s)\tau}$, $\bar{{B}}_{s\tau}$, and $\bar{{A}}_{(1-s)\tau}\bar{{B}}_{s\tau}$, we can upper bound them as
\begin{align}
	&\Big|\sum_k\sum_{n,m}\pi_n\qty(\mel{n}{\bar{{A}}_{(1-s)\tau}}{n}\mel{n}{L_k^\dagger}{m}\mel{m}{\bar{{B}}_{s\tau}}{m}\mel{m}{L_k}{n}-\mel{n}{L_k^\dagger}{m}\mel{m}{\bar{{A}}_{(1-s)\tau}}{m}\mel{m}{L_k}{n}\mel{n}{\bar{{B}}_{s\tau}}{n})\notag\\
	&+\sum_k\sum_{n,m}\pi_n\qty(\mel{n}{L_k^\dagger}{m}\mel{m}{L_k}{n}\mel{n}{\bar{{A}}_{(1-s)\tau}\bar{{B}}_{s\tau}}{n}-\mel{n}{\bar{{A}}_{(1-s)\tau}\bar{{B}}_{s\tau}}{n}\mel{n}{L_k^\dagger}{m}\mel{m}{L_k}{n})/2\Big|\notag\\
	&=\Big|\sum_k\sum_{n,m}\pi_n\qty(\mel{n}{\bar{{A}}_{(1-s)\tau}}{n}\mel{m}{\bar{{B}}_{s\tau}}{m}w_{mn}^k-\mel{m}{\bar{{A}}_{(1-s)\tau}}{m}\mel{n}{\bar{{B}}_{s\tau}}{n}w_{mn}^k)\Big|\notag\\
	&=\Big|\sum_k\sum_{n,m}\mel{n}{\bar{{A}}_{(1-s)\tau}}{n}\mel{m}{\bar{{B}}_{s\tau}}{m}\qty(w_{mn}^k\pi_n-w_{nm}^{k'}\pi_m)\Big|\notag\\
	&\le\max_{m,n}|\mel{n}{\bar{{A}}_{(1-s)\tau}}{n}\mel{m}{\bar{{B}}_{s\tau}}{m}|\sum_k\sum_{n,m}|w_{mn}^k\pi_n-w_{nm}^{k'}\pi_m|\notag\\
	&\le \|\bar{{A}}_{(1-s)\tau}\|_\infty\|\bar{{B}}_{s\tau}\|_\infty\sum_k\sum_{m,n}|j_{mn}^k|\notag\\
	&\le e^{-g\tau}\|A\|_*\|B\|_*\sigma\Phi\qty(\sigma/2\gamma)^{-1}.\label{eq:qbound.tmp2}
\end{align}
For the terms that involve the non-diagonal elements of $\bar{{A}}_{(1-s)\tau}$, $\bar{{B}}_{s\tau}$, and $\bar{{A}}_{(1-s)\tau}\bar{{B}}_{s\tau}$, we can evaluate them as follows:
\begin{align}
	\Big|\sum_{k,n,m(\neq n)}\pi_n\mel{n}{\bar{{A}}_{(1-s)\tau}}{m}\mel{m}{L_k^\dagger\bar{{B}}_{s\tau}L_k}{n}\Big|&\le \|\bar{B}_{s\tau}\|_\infty\sum_k\|L_k\|_\infty^2\sum_{n\neq m}\pi_n|\mel{n}{\bar{{A}}_{(1-s)\tau}}{m}|\notag\\
	&\le e^{-g\tau}\|B\|_*\sum_k\|L_k\|_\infty^2C_{\ell_1}(A_{(1-s)\tau}),\\
	\Big|\sum_{k,n,m,m'(\neq m)}\pi_n\mel{n}{\bar{{A}}_{(1-s)\tau}}{n}\mel{n}{L_k^\dagger}{m}\mel{m}{\bar{{B}}_{s\tau}}{m'}\mel{m'}{L_k}{n}\Big|&\le \|\bar{A}_{(1-s)\tau}\|_\infty\sum_k\|L_k\|_\infty^2\sum_{n\neq m}|\mel{n}{\bar{{B}}_{s\tau}}{m}|\notag\\
	&\le e^{-g\tau}\|A\|_*\sum_k\|L_k\|_\infty^2C_{\ell_1}(B_{s\tau}),\\
	\Big|\sum_{k,n,m(\neq n)}\pi_n\mel{n}{L_k^\dagger\bar{{A}}_{(1-s)\tau}L_k}{m}\mel{m}{\bar{{B}}_{s\tau}}{n}\Big|&\le \|\bar{A}_{(1-s)\tau}\|_\infty\sum_k\|L_k\|_\infty^2\sum_{n\neq m}\pi_n|\mel{n}{\bar{{B}}_{s\tau}}{m}|\notag\\
	&\le e^{-g\tau}\|A\|_*\sum_k\|L_k\|_\infty^2C_{\ell_1}(B_{s\tau}),\\
	\Big|\sum_{k,n,m,m'(\neq m)}\pi_n\mel{n}{L_k^\dagger}{m'}\mel{m'}{\bar{{A}}_{(1-s)\tau}}{m}\mel{m}{L_k}{n}\mel{n}{\bar{{B}}_{s\tau}}{n}\Big|&\le \|\bar{B}_{s\tau}\|_\infty\sum_k\|L_k\|_\infty^2\sum_{n\neq m}|\mel{n}{\bar{{A}}_{(1-s)\tau}}{m}|\notag\\
	&\le e^{-g\tau}\|B\|_*\sum_k\|L_k\|_\infty^2C_{\ell_1}(A_{(1-s)\tau}),
\end{align}
and
\begin{align}
	&\Big|\sum_{k,n,m(\neq n)}\pi_n\mel{n}{\bar{{A}}_{(1-s)\tau}\bar{{B}}_{s\tau}}{m}\mel{m}{L_k^\dagger L_k}{n}\Big|\notag\\
	&\le \sum_{k}\|L_k\|_\infty^2\sum_{n\neq m}\pi_n|\mel{n}{\bar{{A}}_{(1-s)\tau}\bar{{B}}_{s\tau}}{m}|\notag\\
	&\le \sum_{k}\|L_k\|_\infty^2\sum_{n\neq m}\Big[\pi_n|\mel{n}{\bar{{A}}_{(1-s)\tau}}{m}\mel{m}{\bar{{B}}_{s\tau}}{m}|+\sum_{m'(\neq m)}\pi_n|\mel{n}{\bar{{A}}_{(1-s)\tau}}{m'}\mel{m'}{\bar{{B}}_{s\tau}}{m}|\Big]\notag\\
	&\le e^{-g\tau}\sum_{k}\|L_k\|_\infty^2\qty[\|B\|_*C_{\ell_1}(A_{(1-s)\tau})+\|A\|_*C_{\ell_1}(B_{s\tau})].
\end{align}
By combining all these inequalities, the following upper bound for the second term in Eq.~\eqref{eq:qbound.tmp1} is immediately derived:
\begin{align}
	&\Big|\sum_k\ev{\bar{{B}}_{s\tau},L_k\pi\bar{{A}}_{(1-s)\tau}L_k^\dagger-\pi L_k^\dagger\bar{{A}}_{(1-s)\tau}L_k+(\pi L_k^\dagger L_kA_{(1-s)\tau}-L_k^\dagger L_k\pi\bar{{A}}_{(1-s)\tau})/2}\Big|\notag\\
	&\le 3e^{-g\tau}\sum_k\|L_k\|_\infty^2\qty[\|B\|_*C_{\ell_1}(A_{(1-s)\tau})+\|A\|_*C_{\ell_1}(B_{s\tau})]+e^{-g\tau}\|A\|_*\|B\|_*\sigma\Phi\qty(\sigma/2\gamma)^{-1}.\label{eq:qbound.tmp6}
\end{align}
Finally, by inserting Eqs.~\eqref{eq:qbound.tmp7} and \eqref{eq:qbound.tmp6} to Eq.~\eqref{eq:qbound.tmp8}, we obtain the following bound on the asymmetry of cross-correlations:
\begin{equation}
	|\delta C_{ba}^\tau|\le \tau e^{-g\tau}\qty[\mca{C} + \|A\|_*\|B\|_*\sigma\Phi\qty(\frac{\sigma}{2\gamma})^{-1}],
\end{equation}
where $\mca{C}$ is the quantum coherence term given by
\begin{equation}
	\mca{C}\coloneqq (4\|H\|_\infty+3\sum_k\|L_k\|_\infty^2)\int_0^1\dd{s}\qty[ \|B\|_*{C}_{\ell_1}(A_{(1-s)\tau}) + \|A\|_*{C}_{\ell_1}(B_{s\tau}) ].
\end{equation}

\subsection{Analytical demonstration of the critical role of quantum coherence}\label{app:qua.demon}
Here we show that quantum coherence plays a pivotal role in constraining the asymmetry of cross-correlations.
Specifically, we present a case wherein the asymmetry of cross-correlations can persist even when irreversible entropy production is zero.
This implies that the quantum coherence term $\mca{C}$ is inevitable in the derived bound \eqref{eq:qua.asym.cc}.

We consider a three-level maser---the prototype for quantum heat engines that rely on quantum coherence to perform work \cite{Scovil.1959.PRL,Kalaee.2021.PRE,Vu.2022.PRL.TUR}.
The engine is simultaneously coupled to a hot and a cold heat bath and interacts with a classical electric field.
The Markovian dynamics is governed by the local master equation with the Hamiltonian $H_t=H_0+V_t$ and jump operators $L_{1}=\sqrt{\alpha_h(N_h+1)}\sigma_{13}$, $L_{1'}= \sqrt{\alpha_hN_h}\sigma_{31}$, $L_{2}=\sqrt{\alpha_c(N_c+1)}\sigma_{23}$, and $L_{2'}=\sqrt{\alpha_cN_c}\sigma_{32}$.
Here, $H_0=\omega_1\sigma_{11}+\omega_2\sigma_{22}+\omega_3\sigma_{33}$ is the bare Hamiltonian, $V_t=\Omega\qty( e^{i\omega_0t}\sigma_{12}+ e^{-i\omega_0t}\sigma_{21}) $ is the external classical field, $\sigma_{ij}=\dyad{\epsilon_i}{\epsilon_j}$, and $\alpha_x$ and $N_x$ are the decay rate and the thermal occupation number for $x\in\{h,c\}$, respectively.
To remove the time dependence of the full Hamiltonian, it is convenient to rewrite operators in the rotating frame $X\to U_t^\dagger XU_t$, where $U_t=e^{-i\bar{H}t}$ and $\bar{H}=\omega_1\sigma_{11}+(\omega_1+\omega_0)\sigma_{22}+\omega_3\sigma_{33}$.
In this rotating frame, the master equation reads
\begin{equation}\label{eq:rot.frame.Ham}
\dot{{\varrho}}_t = -i[ H,{\varrho}_t ] +\sum_{k=1}^2\qty(\mca{D}[L_{k}]{\varrho}_t+ \mca{D}[L_{k'}]{\varrho}_t),
\end{equation}
where $H=-\Delta \sigma_{22}+\Omega( \sigma_{12}+\sigma_{21} )$ and $\Delta=\omega_0+\omega_1-\omega_2$.
It was shown that the master equation \eqref{eq:rot.frame.Ham} is valid when the driving field is weak \cite{Geva.1996.JCP}.

After some algebraic calculations, we can show that the steady-state density matrix reads
\begin{equation}
	\pi=\pi_{11}\dyad{\epsilon_1}+\pi_{22}\dyad{\epsilon_2}+\pi_{12}\dyad{\epsilon_1}{\epsilon_2}+\pi_{12}^*\dyad{\epsilon_2}{\epsilon_1}+(1-\pi_{11}-\pi_{22})\dyad{\epsilon_3},
\end{equation}
where
\begin{align}
	\pi_{11}&=\mca{F}^{-1}\qty{\alpha_c\alpha_hN_c(N_h+1)[4\Delta^2+(\alpha_cN_c+\alpha_hN_h)^2]+4\Omega^2(\alpha_cN_c+\alpha_hN_h)(\alpha_c+\alpha_h+\alpha_cN_c+\alpha_hN_h)},\\
	\pi_{22}&=\mca{F}^{-1}\qty{\alpha_c\alpha_hN_h(N_c+1)[4\Delta^2+(\alpha_cN_c+\alpha_hN_h)^2]+4\Omega^2(\alpha_cN_c+\alpha_hN_h)(\alpha_c+\alpha_h+\alpha_cN_c+\alpha_hN_h)},\\
	\pi_{12}&=\mca{F}^{-1}\qty{2i\alpha_c\alpha_h\Omega(N_c-N_h)(-2i\Delta +\alpha_cN_c+\alpha_hN_h)},\\
	\mca{F}&=\alpha_c\alpha_h(3N_cN_h+N_c+N_h)[4\Delta^2+(\alpha_cN_c+\alpha_hN_h)^2]+4\Omega^2(\alpha_cN_c+\alpha_hN_h)[\alpha_c(3N_c+2)+\alpha_h(3 N_h+2)].
\end{align}
Likewise, the irreversible entropy production rate is given by
\begin{equation}
	\sigma=4(N_c-N_h)\alpha_c\alpha_h(\alpha_c N_c+\alpha_h N_h)\Omega^2\ln\left(\frac{N_c N_h+N_c}{N_c N_h+N_h}\right).
\end{equation}
Now we consider the simple case of $N_h=N_c$.
In this case, we have $\sigma=0$, $\pi_{11}=\pi_{22}$, and $\pi_{12}=0$, which yield
\begin{equation}
	\pi=\pi_{11}(\dyad{\epsilon_1}+\dyad{\epsilon_2})+(1-2\pi_{11})\dyad{\epsilon_3}.
\end{equation}
For any real numbers $\theta$ and $\phi$, the density matrix $\pi$ can also be written as
\begin{equation}
	\pi=\pi_{11}(\dyad{1}+\dyad{2})+(1-2\pi_{11})\dyad{3},
\end{equation}
where $\ket{1}=e^{i\phi}\cos(\theta)\ket{\epsilon_1}+\sin(\theta)\ket{\epsilon_2}$, $\ket{2}=-\sin(\theta)\ket{\epsilon_1}+e^{-i\phi}\cos(\theta)\ket{\epsilon_2}$, and $\ket{3}=\ket{\epsilon_3}$.
We consider observables $A=\dyad{1}$ and $B=\dyad{3}$.
We need only show that the asymmetry of cross-correlations is nonzero for this measurement basis.
It is thus sufficient to prove for the short-time regime $\tau\ll 1$.
In this region of operational time $\tau$, the asymmetry of cross-correlations can be analytically expanded in terms of $\tau$ as
\begin{align}
	\delta C_{ba}^\tau&=\tr{B[\mbb{1}+\tau\mca{L}+\tau^2\mca{L}^2/2](A\pi)-B\pi[\mbb{1}+\tau\tilde{\mca{L}}+\tau^2\tilde{\mca{L}}^2/2](A)}+O(\tau^3)\notag\\
	&=\tr{B[\mca{L}(A\pi)-\pi\tilde{\mca{L}}(A)]}\tau+\tr{B[\mca{L}^2(A\pi)-\pi\tilde{\mca{L}}^2(A)]}\frac{\tau^2}{2}+O(\tau^3)\notag\\
	&=\frac{N_h(N_h+1)\Omega(\alpha_c-\alpha_h)\sin(2\theta)\sin(\phi)}{3N_h+2}\tau^2+O(\tau^3).
\end{align}
As can be observed, although the first-order term is zero, the second-order term does not vanish and $\delta C_{ba}^\tau$ is thus nonzero (i.e., $|\delta C_{ba}^\tau|>0$).
Since $\sigma=0$, the quantum coherence term $\mca{C}$ is the only term that bounds $\delta C_{ba}^\tau$ in this case:
\begin{equation}
	|\delta C_{ba}^\tau|\le \tau e^{-g\tau}\mca{C}.
\end{equation}

\section{Proofs of Eqs.~\eqref{eq:asym.cc.fnt.2} and \eqref{eq:asym.cc.fnt.aff}}

\subsection{Proof of Eq.~\eqref{eq:asym.cc.fnt.2}}\label{app:main2.proof}
Define ${T}\coloneqq e^{{W}\tau}\Pi-\Pi$, which satisfies ${T}\ket{1}=\ket{0}$ and ${T}^\dagger\ket{1}=\ket{0}$.
In other words, $T_{nn}=-\sum_{m(\neq n)}T_{mn}=-\sum_{m(\neq n)}T_{nm}=-\sum_{m(\neq n)}(T_{mn}+T_{nm})/2$.
For $m\neq n$, we have $T_{mn}=[e^{W\tau}]_{mn}\pi_n\ge 0$, which is nothing but the joint probability of observing the initial and final states.
First, we prove that
\begin{equation}\label{eq:cg.ent.prod}
	\tau\sigma\ge D({T}||{T}^\dagger)=\sum_{m,n}T_{mn}\ln\frac{T_{mn}}{T_{nm}}.
\end{equation}
To this end, let $\Gamma$ be a stochastic trajectory of system states and $p(\Gamma)$ be the path probability of finding $\Gamma$.
Utilizing the phase-space representation of the total entropy production and the monotonicity of the Kullback-Leibler divergence under coarse-graining, Eq.~\eqref{eq:cg.ent.prod} can be proved as
\begin{align}
	\tau\sigma=D(p(\Gamma)||p(\tilde{\Gamma}))&\ge D(\Lambda[p(\Gamma)]||\Lambda[p(\tilde{\Gamma})])\notag\\
	&=D({T}||{T}^\dagger),
\end{align}
where $\Lambda$ is a coarse-grained map that reduces the path probability to the probability of observing only the initial and final states.
For convenience, we define $\ell_{mn}\coloneqq\sqrt{(a_m-a_n)^2+(b_m-b_n)^2}$, $\Omega_{mn}\coloneqq(a_nb_m-a_mb_n)/2$, and $J_{mn}\coloneqq T_{mn}-T_{nm}$.
Note that $J_{mn}$ differs from $j_{mn}$.
Using these quantities, we can express the asymmetry of cross-correlations as
\begin{align}
    \delta C_{ba}^\tau=\mel{b}{{T}-{T}^\dagger}{a}&=\sum_{m>n}(T_{mn}-T_{nm})(a_nb_m-a_mb_n)\notag\\
    &=2\sum_{m>n}J_{mn}\Omega_{mn}.\label{eq:asym.cc.tmp6}
\end{align}
Likewise, we can calculate the decay of auto-correlation as
\begin{align}
	D_a^\tau=-\mel{a}{{T}}{a}&=\sum_{m\neq n}\qty(\frac{T_{mn}+T_{nm}}{2}a_n^2-T_{mn}a_ma_n)\notag\\
	&=\frac{1}{2}\sum_{m>n}(T_{mn}+T_{nm})(a_m-a_n)^2,
\end{align}
which leads to
\begin{align}
	D_a^\tau+D_b^\tau&=\frac{1}{2}\sum_{m>n}(T_{mn}+T_{nm})\ell_{mn}^2.\label{eq:asym.cc.tmp7}
\end{align}
Using Eqs.~\eqref{eq:asym.cc.tmp6} and \eqref{eq:asym.cc.tmp7}, we can prove the first argument of Eq.~\eqref{eq:asym.cc.fnt.2} as follows:
\begin{align}
	\frac{|\delta C_{ba}^\tau|^2}{D_a^\tau+D_b^\tau}&=\frac{8\qty(\sum_{m>n}J_{mn}\Omega_{mn})^2}{\sum_{m>n}(T_{mn}+T_{nm})\ell_{mn}^2}\notag\\
	&\le \frac{2\|a^2+b^2\|_\infty\qty(\sum_{m>n}|J_{mn}|\ell_{mn})^2}{\sum_{m>n}(T_{mn}+T_{nm})\ell_{mn}^2}\notag\\
	&\le 2\|a^2+b^2\|_\infty\sum_{m>n}\frac{\qty(T_{mn}-T_{nm})^2}{T_{mn}+T_{nm}}\notag\\
	&\le \|a^2+b^2\|_\infty\sum_{m>n}(T_{mn}-T_{nm})\ln\frac{T_{mn}}{T_{nm}}\notag\\
	&=\|a^2+b^2\|_\infty D({T}||{T}^\dagger)\notag\\
	&\le \|a^2+b^2\|_\infty\tau\sigma.\label{eq:asym.cc.tmp8}
\end{align}
Here, we use the triangle inequality and $|\Omega_{mn}|\le \ell_{mn}\|a^2+b^2\|_\infty^{1/2}/2$ \cite{Shiraishi.2023.PRE} in the second line, the Cauchy-Schwarz inequality in the third line, inequality $2(x-y)^2/(x+y)\le (x-y)\ln(x/y)$ in the fourth line, and Eq.~\eqref{eq:cg.ent.prod} in the last line.

The second argument of Eq.~\eqref{eq:asym.cc.fnt.2} can be proved by a similar strategy as in Ref.~\cite{Shiraishi.2023.PRE}.
For any directed edge $e=(m\leftarrow n)$, we define $x_e\coloneqq x_{mn}$ for an arbitrary variable $x$, its reversed edge $\tilde{e}\coloneqq(n\leftarrow m)$, and $\mca{E}\coloneqq\{e\,|\,J_e>0\}$.
The discrete isoperimetric inequality \cite{Ohga.2023.PRL,Fan.1955.JWAS} implies that
\begin{equation}
	4|c|\tan\frac{\pi}{|c|}|\Omega_c|\le \ell_c^2,\label{eq:iso.ine}
\end{equation}
where we define $X_c\coloneqq\sum_{e\in c}X_e$ for variable $X$ and cycle $c\in\mca{C}$.
Since $\sum_{m}J_{mn}=0$, we can always find a uniform decomposition of cycles $\mca{C}$ with appropriate orientations and associated positive currents $\{J^c\}_{c\in\mca{C}}$ such that $J_e=\sum_cJ^cS_{e}^c$ for any $e\in\mca{E}$, where $S_{e}^c=1$ if $e\in c$ and zero otherwise \cite{Pietzonka.2016.JPA}.
Using this decomposition, equality $\sum_{c\in\mca{C}}J^cX_c=\sum_{e\in\mca{E}}J_eX_e$ can be derived; thus, $\delta C_{ba}^\tau=2\sum_{e\in\mca{E}}J_e\Omega_e=2\sum_{c\in\mca{C}}J^c\Omega_c$.
By utilizing this equality, inequality \eqref{eq:iso.ine}, and the monotonicity of function $x\tan(\pi/x)$ over $[3,\infty)$, the asymmetry of cross-correlations can be upper bounded as
\begin{align}
	|\delta C_{ba}^\tau|&\le\sum_{c\in\mca{C}}J^c\ell_c^2\qty(2|c|\tan\frac{\pi}{|c|})^{-1}\label{eq:asym.cc.tmp9}\\
	&\le\qty(2N\tan\frac{\pi}{N})^{-1}\max_c\ell_c\sum_{m>n}|J_{mn}|\ell_{mn}.
\end{align}
Subsequently, following the same procedure as in Eq.~\eqref{eq:asym.cc.tmp8} leads to the desired result.

\subsection{Proof of Eq.~\eqref{eq:asym.cc.fnt.aff}}\label{app:main3.proof}
We follow the approach in Ref.~\cite{Ohga.2023.PRL}.
Note that observables $a$ and $b$ can be arbitrarily rescaled without altering the ratio $|\delta C_{ba}^\tau|/\sqrt{D_a^\tau D_b^\tau}$.
Therefore, we can assume $D_a^\tau=D_b^\tau$ without loss of generality.
Noticing that $\mca{F}_c^\tau=\sum_{e\in c}\ln(T_e/T_{\tilde{e}})$ and applying Jensen's inequality, we can lower bound $\mca{F}_c^\tau$ as
\begin{align}
	\mca{F}_c^\tau&=\sum_{e\in c}\ln\frac{(T_e+T_{\tilde{e}})+J_e}{(T_e+T_{\tilde{e}})-J_e}\notag\\
	&=2|c|\sum_{e\in c}\frac{1}{|c|}\atanh\qty(\frac{J_e}{T_e+T_{\tilde{e}}})\notag\\
	&\ge 2|c|\atanh\qty(\frac{1}{|c|}\sum_{e\in c}\frac{J_e}{T_e+T_{\tilde{e}}}),
\end{align}
which yields
\begin{equation}
	\frac{1}{|c|}\sum_{e\in c}\frac{J_e}{T_e+T_{\tilde{e}}}\le\tanh\qty(\frac{\mca{F}_c^\tau}{2|c|}).
\end{equation}
Consequently, by applying the Cauchy-Schwarz inequality, we obtain
\begin{align}
	\sum_{e\in c}\frac{(T_e+T_{\tilde{e}})}{J_e}\ell_e^2&\ge\frac{\qty(\sum_{e\in c}\ell_e)^2}{\sum_{e\in c}J_e/(T_e+T_{\tilde{e}})} \notag\\
	&\ge\frac{\ell_c^2}{|c|\tanh(\mca{F}_c^\tau/2|c|)}.
\end{align}
Using this inequality, we can lower bound the denominator as 
\begin{align}
	2\sqrt{D_a^\tau D_b^\tau}&=D_a^\tau+D_b^\tau\ge\frac{1}{2}\sum_{e\in\mca{E}}(T_e+T_{\tilde{e}})\ell_e^2\notag\\
	&=\frac{1}{2}\sum_{c\in\mca{C}}J^c\sum_{e\in c}\frac{(T_e+T_{\tilde{e}})}{J_e}\ell_e^2\notag\\
	&\ge \frac{1}{2}\sum_{c\in\mca{C}}J^c\ell_c^2\qty[|c|\tanh(\mca{F}_c^\tau/2|c|)]^{-1}.\label{eq:asym.cc.tmp10}
\end{align}
By combining Eqs.~\eqref{eq:asym.cc.tmp9} and \eqref{eq:asym.cc.tmp10} and noticing that $(\sum_cx_c)/(\sum_cy_c)\le\max_c(x_c/y_c)$ for positive numbers $\{x_c\}$ and $\{y_c\}$, Eq.~\eqref{eq:asym.cc.fnt.aff} is immediately derived.

\subsubsection{Analytical demonstration of the bound's attainability}
Here we analytically demonstrate that the equality of the bound \eqref{eq:asym.cc.fnt.aff} can be attained for arbitrary times in the three-state biochemical oscillation with homogeneous transition rates.
Without loss of generality, we can assume that $w_+=\omega>1$ and $w_-=1$.
Note that the observables are given by $\ket{a}=[\sin(2\pi n/3)]_n^\top$ and $\ket{b}=[\cos(2\pi n/3)]_n^\top$.
For this unicyclic system, by performing simple algebraic calculations, we can derive that
\begin{align}
	\frac{|\delta C_{ba}^\tau|}{2\sqrt{D_a^\tau D_b^\tau}}&=\frac{\qty|\sin[\sqrt{3}(\omega-1)\tau/2]|}{e^{3(\omega+1)\tau/2}-\cos[\sqrt{3}(\omega-1)\tau/2]},\\
	\max_c\frac{\tanh(\mca{F}_c^\tau/2|c|)}{\tan(\pi/|c|)}&=\frac{1}{\sqrt{3}}\tanh\qty(\frac{1}{2}\qty|\log\frac{e^{3(\omega+1)\tau/2}-\cos[\sqrt{3}(\omega-1)\tau/2]+\sqrt{3}\sin[\sqrt{3}(\omega-1)\tau/2]}{e^{3(\omega+1)\tau/2}-\cos[\sqrt{3}(\omega-1)\tau/2]-\sqrt{3}\sin[\sqrt{3}(\omega-1)\tau/2]}|).
\end{align}
By separately considering two cases: $\sin[\sqrt{3}(\omega-1)\tau/2]\ge 0$ and $\sin[\sqrt{3}(\omega-1)\tau/2]<0$, one can easily verify the equality of the bound: 
\begin{equation}
	\frac{|\delta C_{ba}^\tau|}{2\sqrt{D_a^\tau D_b^\tau}}=\max_c\frac{\tanh(\mca{F}_c^\tau/2|c|)}{\tan(\pi/|c|)}.
\end{equation}

\section{Useful inequalities}\label{app:useful.ines}

\begin{proposition}\label{prop:tr.prod.ine}
For any matrices ${X}$ and ${Y}$ and orthogonal basis $\{\ket{n}\}$, we have
\begin{equation}
    \qty|\tr{{X}{Y}}|\le\|{Y}\|_\infty\sum_{m,n}|\mel{m}{{X}}{n}|,
\end{equation}
where $\|{Y}\|_\infty$ denotes the operator norm of ${Y}$.
\end{proposition}
\begin{proof}
By applying the inequalities $\tr{{X}{Y}}\le\|{Y}\|_\infty\|{X}\|_1$ and $\|{X}\|_1\le\sum_{m,n}|\mel{m}{{X}}{n}|$, one can immediately complete the proof.
\end{proof}

\begin{proposition}\label{prop:z.diff.ine}
For any complex numbers $z_1$ and $z_2$ with a negative real part, the following inequality holds:
\begin{equation}\label{eq:z.ine.tmp1}
    \qty|\frac{e^{z_1}-e^{z_2}}{z_1-z_2}|\le 1.
\end{equation}
\end{proposition}
\begin{proof}
Since $\Re{z_1}\le 0$ and $\Re{z_2}\le 0$, we have $|e^{sz_1+(1-s)z_2}|\le 1$ for any $0\le s\le 1$.
Consequently, we can prove Eq.~\eqref{eq:z.ine.tmp1} as follows:
\begin{align}
	\qty|\frac{e^{z_1}-e^{z_2}}{z_1-z_2}|&=\qty|\int_0^1\dd{s}e^{sz_1+(1-s)z_2}|\notag\\
	&\le \int_0^1\dd{s}\qty|e^{sz_1+(1-s)z_2}|\notag\\
	&\le\int_0^1\dd{s}\notag\\
	&=1. 
\end{align}
\end{proof}

\twocolumngrid


\begin{thebibliography}{77}
\bibitem [{fnt({\natexlab{a}})}]{fnt1}%
  \BibitemOpen
  \href@noop {} {}Note that the dual dynamics has the same steady state as the
  original one. Thus, the difference between cross-correlations in these
  dynamics is quantified by
  $\mel{b}{e^{{W}\tau}\Pi}{a}-\mel{b}{e^{\widetilde{{W}}\tau}\Pi}{a}=\mel{b}{e^{{W}\tau}\Pi}{a}-\mel{b}{\Pi
  e^{{W}^\dagger\tau}}{a}=\delta C_{ba}^\tau$, which is nothing but the
  asymmetry of cross-correlations.\BibitemShut {Stop}%
\bibitem [{fnt({\natexlab{c}})}]{fnt3}%
  \BibitemOpen
  \href@noop {} {}Noting that $A=\sum_ka_k\dyad{k}$, quantum coherence
  associated with operator $A_t$ in the Heisenberg picture can be expressed as
  $C_{\ell_1}(A_t)\propto\sum_{m\neq
  n}|\sum_ka_k\mel{m}{e^{\tilde{\mca{L}}t}(\dyad{k})}{n}|$. For initial
  incoherent state $\dyad{k}$, the term
  $\mel{m}{e^{\tilde{\mca{L}}t}(\dyad{k})}{n}$ quantifies the amount of quantum
  coherence generated in the off-diagonal element $\dyad{m}{n}$ after time $t$.
  Therefore, $C_{\ell_1}(A_t)$ can be interpreted as a weighted sum of quantum
  coherence that is created with respect to initial incoherent states
  $\{\dyad{k}\}$ in the Heisenberg picture.\BibitemShut {Stop}%
\end{thebibliography}
\end{document}